\documentclass{article}
\usepackage{fullpage, amsmath, amsthm, amssymb, setspace, epsfig,color}

\newtheorem{theorem}{Theorem}
\newtheorem{lemma}[theorem]{Lemma}
\newtheorem{definition}[theorem]{Definition}

\newtheorem{proposition}[theorem]{Proposition}
\newtheorem{conjecture}[theorem]{Conjecture}

\numberwithin{theorem}{section}

\newtheorem*{lem:maingeneral}{Lemma \ref{LemmaMainGeneral}}
\newtheorem*{thm:prggeneral}{Theorem \ref{TheoremPRGgeneral}}
\newtheorem*{lem:brry}{Lemma \ref{LemmaBRRY}}

\newcommand{\ex}[2]{\underset{#1}{\mathbb{E}}\left[ #2 \right]}
\newcommand{\pr}[2]{\underset{#1}{\mathbb{P}}\left[ #2 \right]}
\newcommand{\norm}[1]{\left|\left| #1 \right|\right|}
\newcommand{\frob}[1]{\norm{#1}_\text{Fr}}
\newcommand{\omitted}[1]{}

\newcommand{\RL}{\textrm{RL}}
\newcommand{\Lspace}{\textrm{L}}
\newcommand{\poly}{\mathrm{poly}}

\newcommand{\eps}{\varepsilon}
\newcommand{\zo}{\{0,1\}}
\newcommand{\tO}{\tilde{O}}
\newcommand{\Select}{\mathrm{Select}}

\newcommand{\Z}{\mathbb{Z}}
\newcommand{\N}{\mathbb{N}}

\title{Pseudorandomness for Regular Branching Programs\\ via Fourier Analysis}
\author{Omer Reingold\thanks{Microsoft Research Silicon Valley, 1065 La Avenida, Mountain View, CA.	}\\\small\texttt{Omer.Reingold@microsoft.com} \and  Thomas Steinke\thanks{School of Engineering and Applied Sciences, Harvard University, 33 Oxford Street, Cambridge MA. Work done in part while at Stanford University. Supported by NSF grant CCF-1116616 and the Lord Rutherford Memorial Research Fellowship.}\\\small\texttt{tsteinke@seas.harvard.edu} \and Salil Vadhan\thanks{School of Engineering and Applied Sciences, Harvard University, 33 Oxford Street, Cambridge MA.  Work
done in part when on leave as a Visiting Researcher at Microsoft Research Silicon Valley and a Visiting Scholar at Stanford University.  Supported in part by NSF grant CCF-1116616 and US-Israel BSF grant 2010196.} 
\\\small\texttt{salil@seas.harvard.edu}}

\date{19 June 2013}

\begin{document}

\begin{titlepage}

\maketitle
\thispagestyle{empty}

\begin{abstract}
We present an explicit pseudorandom generator for oblivious, read-once, permutation branching programs of constant width that can read their input bits in any order. The seed length is $O(\log^2 n)$, where $n$ is the length of the branching program. The previous best seed length known for this model was $n^{1/2+o(1)}$, which follows as a special
case of a generator due to Impagliazzo, Meka, and Zuckerman (FOCS 2012) (which gives a seed length of
$s^{1/2+o(1)}$ for arbitrary branching programs of size $s$). Our techniques also give seed length $n^{1/2+o(1)}$ for general oblivious, read-once branching programs of width $2^{n^{o(1)}}$, which is incomparable to the results of Impagliazzo et al.

Our pseudorandom generator is similar to the one used by Gopalan et al.~(FOCS 2012) for read-once CNFs, but the analysis is quite different; ours is based on Fourier analysis of branching programs. In particular, we show that an oblivious, read-once, \emph{regular} branching program of width $w$ has Fourier mass at most $(2w^2)^{k}$ at level $k$, independent of the length of the program.  
\end{abstract}
\end{titlepage}

\section{Introduction} \label{SectionIntroduction}

A major open problem in the theory of pseudorandomness is to construct an ``optimal'' pseudorandom generator for space-bounded computation.  That is,
we want an
explicit pseudorandom generator that stretches a uniformly random seed
of length $O(\log n)$ to $n$ bits that cannot be distinguished from uniform by any $O(\log n)$-space algorithm (which receives the pseudorandom bits one at a time, in a streaming fashion, and may be nonuniform).

Such a generator would imply that every randomized algorithm can be derandomized with only a constant-factor increase in space ($\RL=\Lspace$), and would also have a variety of other applications, such as
in streaming
algorithms~\cite{Indyk},
deterministic dimension reduction and SDP rounding~\cite{Sivakumar}, 
 hashing~\cite{CReingoldSW}, hardness
amplification~\cite{HealyVaVi}, almost $k$-wise independent
permutations~\cite{KaplanNaRe}, and cryptographic pseudorandom generator constructions~\cite{HaitnerHaRe06}.

Unfortunately, for fooling general logspace algorithms, there has been essentially no improvement since the classic work of Nisan~\cite{Nisan}, which provided a pseudorandom generator of seed length $O(\log ^2 n)$.  Instead, a variety of works have improved the seed
length for various restricted classes of logspace algorithms, such as
algorithms that use $n^{o(1)}$ random bits~\cite{NZ,RR}, combinatorial rectangles~\cite{EvenGLNV,LinialLSZ,ArmoniSWZ,Lu}
random walks on graphs~\cite{Reingold,RTV}, branching programs of width 2 or 3~\cite{SZ,BogdanovDvVeYe09,SimaZak}, and
regular or permutation branching programs (of bounded width)~\cite{BRRY,BrodyVerbin,KNP,De,Steinke12}.

The vast majority of these works are based on Nisan's generator or its variants by Impagliazzo, Nisan, and Wigderson~\cite{INW} and
Nisan and Zuckerman~\cite{NZ}, and show how the analysis (and hence the final parameters) of these generators can be improved for logspace algorithms that satisfy the additional restrictions.  All three of these generators are based on recursive use of the following principle: if we consider two consecutive time intervals $I_1$, $I_2$ in a space $s$ computation and use some randomness $r$ to generate the pseudorandom bits fed to the algorithm during interval $I_1$, then at the start of $I_2$, the algorithm will `remember' at most $s$ bits of information about $r$.  So we can use a randomness extractor to extract roughly $|r|-s$ almost uniform bits from $r$ (while investing only a small additional amount of randomness for the extraction).  This paradigm seems unlikely to yield pseudorandom generators for general logspace computations that have a seed length of $\log^{1.99} n$ (see \cite{BrodyVerbin}). 

Thus, there is a real need for a different approach to constructing pseudorandom generators for space-bounded computation.  One new approach has been suggested in the recent work of Gopalan et al.~\cite{GMRTV}, which constructed improved pseudorandom generators for read-once CNF formulas and combinatorial rectangles, and hitting set generators for width 3 branching programs.  Their basic generator (e.g. for read-once CNF formulas) works as follows: Instead of considering a fixed partition of the bits into intervals, they pseudorandomly partition the bits into two groups, assign the bits in one group using a small-bias generator~\cite{NaorNa93}, and then recursively generate bits for the second group.  While it would not work to assign {\em all} the bits using a single sample from a small-bias generator, it turns out that generating a pseudorandom partial assignment is a significantly easier task.

An added feature of the Gopalan et al.~generator is that its pseudorandomness properties are independent of the order in which the output bits are read by a potential distinguisher.  In contrast, Nisan's generator and its variants depend heavily on the ordering of bits (the intervals $I_1$ and $I_2$ above cannot be interleaved), and in fact it is known that a particular instantiation
of Nisan's generator fails to be pseudorandom if the (space-bounded) distinguisher can read the bits in a different order~\cite[Corollary 3.18]{Tzur}.  Recent works \cite{BPW,IMZ} have constructed nontrivial pseudorandom generators for space-bounded algorithms that can read their bits in any order, but the seed length achieved is larger than $\sqrt{n}$.

In light of the above, a natural question is whether the approach of Gopalan et al.~can be extended to a wider class of space-bounded algorithms.  We make progress on this question by using the same approach to construct a pseudorandom generator with seed length $O(\log^2 n)$ for constant-width, read-once, oblivious permutation branching programs that can read their bits in any order.  In analysing our generator, we develop new Fourier-analytic tools for proving pseudorandomness against space-bounded algorithms.

\subsection{Models of Space-Bounded Computation}

A (layered) \textbf{branching program} $B$ is a nonuniform model of space-bounded computation.  The program maintains a \textbf{state} from the set $[w]=\{1,\ldots,w\}$ and, at each time step $i$, reads one bit of its input $x\in \zo^n$ and updates its state according to a transition function $B_i : \zo\times [w]\rightarrow [w]$.  The parameter $w$ is called the \textbf{width} of the program, and corresponds to a space bound of $\log w$ bits.
We allow the transition function $B_i$ to be different at each time step $i$.  We consider several restricted forms of branching programs:
\begin{itemize}
\item \textbf{Read-once branching programs} read each input bit at most once.
\item \textbf{Oblivious branching programs} choose which input bit to read depending only on the time step $i$, and not on the current state
\item \textbf{Ordered branching programs} (a.k.a. streaming algorithms) always read input bit $i$ in time step $i$ (hence are necessarily both read-once and oblivious).
\end{itemize}
To derandomize randomized space-bounded computations (e.g. prove $\RL=\Lspace$), it suffices to construct pseudorandom generators that fool ordered branching programs of polynomial width ( $w=\poly(n)$), and hence this is the model addressed by most previous constructions (including Nisan's generator).  However, the more general models of oblivious and read-once branching programs are also natural to study, and, as discussed above, can spark the development of new techniques for reasoning about pseudorandomness.

As mentioned earlier, Nisan's pseudorandom generator~\cite{Nisan} achieves $O(\log^2 n)$ seed length for {\em ordered} branching programs of polynomial width. It is known how to achieve $O(\log n)$ seed length for ordered branching programs width 2~\cite{BogdanovDvVeYe09}, and for width 3, it is only known how to construct ``hitting-set generators'' (a weaker form of pseudorandom generators) with seed length
$O(\log n)$~\cite{SimaZak,GMRTV}.  (The seed length is $\tO(\log n)$ if we want the error of the hitting set generator to be subconstant.)  For pseudorandom generators for width $w\geq 3$ and hitting-set generators for width $w\geq 4$, there is no known construction with seed length $o(\log^2 n)$.

The study of pseudorandomness against non-ordered branching programs started more recently.
Tzur~\cite{Tzur} showed that there are oblivious, read-once, constant-width branching programs that can distinguish the output of Nisan's generator from uniform.  Bogdanov, Papakonstantinou, and Wan~\cite{BPW} exhibited a pseudorandom generator with seed length $(1-\Omega(1))\cdot n$ for oblivious read-once branching programs of width $w$ for $w=2^{\Omega(n)}$.   Impagliazzo, Meka, and Zuckerman~\cite{IMZ} gave a pseudorandom generator with seed length
$s^{1/2+o(1)}$ for arbitrary branching programs of size $s$; note that $s=O(nw)$ for a read-once branching program of width $w$ and length $n$.

We consider two further restrictions on branching programs:
\begin{itemize}
\item \textbf{Regular branching programs} are oblivious branching programs with the property that, if the distribution on states in any layer is uniformly random and the input bit read by the program at that layer is uniformly random, then the resulting distribution on states in the next layer is uniformly random.  This is equivalent to requiring that the bipartite graph associated with each layer of the program, where we have edges from each state $u\in [w]$ in layer $i$ to the possible next-states $u_0,u_1\in [w]$ in layer $i+1$ (if the input bit is $b$, the state goes to $u_b$), is a regular graph.

\item \textbf{Permutation branching programs} are a further restriction, where we require that for each setting of the input string, the mappings between layers are permutations.  This is equivalent to saying that (regular) bipartite graphs corresponding to each layer are decomposed into two perfect matchings, one corresponding to each value of the current input bit being read.
\end{itemize}

The fact that pseudorandomness for permutation branching programs might be easier than for general branching programs was suggested by the proof that Undirected S-T Connectivity is in Logspace~\cite{Reingold} and its follow-ups~\cite{RTV,RV}.  Specifically, the latter works construct ``pseudorandom walk generators'' for ``consistently labelled'' graphs.  Interpreted for permutation branching programs, these results ensure that if an ordered permutation branching program has the property that every layer has a nonnegligible amount of ``mixing'' --- meaning that the distribution on states becomes closer to uniform, on a truly random input --- then the overall program will also have mixing when run on the output of the pseudorandom generator (albeit at a slower rate).  The generator has a seed length of $O(\log n)$ even for ordered permutation branching programs of width $\poly(n)$.
Reingold, Trevisan, and Vadhan~\cite{RTV} also show that if a generator with similar properties could be constructed for (ordered) regular branching programs of polynomial width, then this would suffice to prove $\RL=\Lspace$.  Thus, in the case of polynomial width, regularity is not a significant constraint.

Recently, there has been substantial progress on constructing pseudorandom generators for ordered regular and permutation branching programs of constant width.
Braverman, Rao, Raz, and Yehudayoff~\cite{BRRY} and Brody and Verbin~\cite{BrodyVerbin} gave pseudorandom generators with seed length $\tO(\log n)$ for ordered regular branching programs of constant width.  Kouck\'{y}, Nimbhorkar and Pudl\'{a}k~\cite{KNP} showed that the seed length could be further
improved to $O(\log n)$ for ordered, permutation branching programs of constant width; see \cite{De,Steinke12} for simplifications and improvements.

All of these generators for ordered regular and permutation branching programs are based on refined analyses of the pseudorandom generator construction of
Impagliazzo, Nisan, and Wigderson~\cite{INW}.

\subsection{Our Results and Techniques}

Our main result is a pseudorandom generator for read-once, oblivious, (unordered) permutation branching programs of constant width:

\begin{theorem}[Main Result] \label{thm:main-intro}
For every constant $w$, there is an explicit pseudorandom generator $G : \zo^{O(\log^2 n)} \rightarrow \zo^n$ fooling oblivious, read-once (but unordered), permutation branching programs of width $w$ and length $n$.
\end{theorem}

To be precise, the seed length and space complexity of the pseudorandom generator is $$O(w^2 \log(w) \log(n) \log(nw/\varepsilon) + w^4 \log^2(w/\varepsilon))$$ for oblivious, read-once, permutation branching programs of length $n$ and width $w$, where $\varepsilon$ is the error.

Previously, it was only known how to achieve a seed length of $n^{1/2+o(1)}$ for this model, as follows from the aforementioned results of Impagliazzo, Meka, and Zuckerman~\cite{IMZ} (which actually holds for arbitrary branching programs).

Our techniques also achieve seed length $n^{1/2+o(1)}$ for arbitrary read-once, oblivious branching programs of width up to $2^{n^{o(1)}}$:

\begin{theorem} \label{thm:second-intro}
There is an explicit pseudorandom generator $G : \zo^{\tO(\sqrt{n}\log w)} \rightarrow \zo^n$ fooling oblivious, read-once (but unordered) branching programs of width $w$ and length $n$.
\end{theorem}

This result is incomparable to that of Impagliazzo et al.~\cite{IMZ}. Their seed length depends polynomially on the width $w$, so require width $w=n^{o(1)}$ to achieve seed length $n^{1/2+o(1)}$. On the other hand, our result is restricted to \emph{read-once, oblivious} branching programs.

Our construction of the generator in Theorem \ref{thm:main-intro} is essentially the same as the generator of Gopalan et al.~\cite{GMRTV} for read-once CNF formulas, but with a new
analysis (and different setting of parameters) for read-once, oblivious, permutation branching programs. The generator works by selecting a subset $T\subset [n]$ of output coordinates in a pseudorandom way, assigning
the bits in $T$ using another pseudorandom distribution $X$, and then recursively assigning the bits outside $T$.   We generate $T$ using an almost $O(\log n)$-wise independent distribution, including each coordinate $i\in T$ with a constant probability $p_w$ depending only on the width $w$.  We assign the bits in
$T$ using a small-bias distribution $X$ on $\zo^n$~\cite{NaorNa93}; such a generator has the property that for every nonempty subset $S\subset [n]$, the parity $\oplus_{i\in S} X_i$ of bits in $S$ has bias at most $\eps$.  Generating $T$ requires $O(\log n)$ random bits, generating $X$ requires $O(\log n)$ bits (even for $\eps=1/\poly(n)$), and we need $O(\log n)$ levels of recursion to assign all the bits.  This gives us our $O(\log^2 n)$ seed length.

Let $B : \zo^n\rightarrow \zo$ be a function computed by an oblivious, read-once, permutation branching program of width $w$.  Following \cite{GMRTV}, to show that our pseudorandom generator fools $B$, it suffices to show that the partial assignment generated in a single level of recursion approximately preserves the acceptance probability of $B$ (on average).  To make this precise, we need a bit of notation.  For a set $t\subset [n]$, a string $x\in \zo^n$, and $y\in \zo^{n-|t|}$,
define $\Select(t,x,y)\in \zo^n$ as follows:
$$\Select(t,x,y)_i = \begin{cases}
x_i & \text{if $i\in t$}\\
y_{|\{j\leq i : j\notin t\}|}
& \text{if $i\notin t$}
\end{cases}$$
Once we choose a set $t\leftarrow T$ and an assignment $x\leftarrow X$ to the variables in $t$, the residual acceptance probability of $B$ is $\pr{U}{B(\Select(t,x,U))=1}$, where $U$ is the uniform distribution on $\{0,1\}^n$.  So, the average acceptance probability over $t\leftarrow T$ and $x\leftarrow X$ is $\pr{T,X,U}{B(\Select(T,X,U))=1}$.  We would like this to be close to the acceptance probability under uniformly random bits, namely $\pr{U}{B(U)=1}=\pr{T,U',U}{B(\Select(T,U',U)=1}$.  That is, we would like our small-bias distribution $X$ to fool the function $B'(x) := \ex{T,U}{B(\Select(T,x,U))}$.  The key insight in \cite{GMRTV}
is that $B'$ can be a significantly easier function to fool than $B$, and even than fixed restrictions of $B$ (like
$B(\Select(t,\cdot,y))$ for fixed $t$ and $y$).  We show that the same phenomenon holds for oblivious, read-once, {\em regular} branching programs. (The reason that the analysis of our overall pseudorandom generator applies only for 
{\em permutation} branching programs is that regularity is not preserved under restriction (as needed for the recursion), whereas the permutation property is.)

To show that a small-bias space fools $B'(x)$, it suffices to show that the \textbf{Fourier mass} of $B'$, namely
$\sum_{s\in \zo^n,s\neq 0} |\widehat{B'}[s]|$, is bounded by $\poly(n)$.    (Here
$\widehat{B'}[s] = \ex{U}{B'[U] \cdot (-1)^{s\cdot U}}$ is the standard Fourier transform over $\Z_2^n$. So $\widehat{B'}[s]$ measures the correlation of $B'$ with the parity function defined by $s$.)  We show that this is indeed the case (for most choices of the set $t\leftarrow T$):

\begin{theorem}[Main Lemma] \label{thm:mainlemma-intro}
For every constant $w$, there are constants $p_w>0$ and $d_w\in \N$ such that the following holds. Let
$B : \zo^n\rightarrow \zo$ be computed by an oblivious, read-once, regular branching program of width $w$ and length $n \geq d_w$.  Let
$T\subset [n]$ be a randomly chosen set so that every coordinate $i\in [n]$ is placed in $T$ with probability $p_w$ and these choices
are $n^{-d_w}$-almost $(d_w\log n)$-wise independent.
Then with high probability over $t\leftarrow T$
$B'(x) = \ex{U}{B(\Select(t,x,U))}$ has Fourier mass at most $n^{d_w}$.
\end{theorem}

As a warm-up, we begin by analysing the Fourier mass
in the case the set $T$ is chosen completely at random, with every coordinate included independently with probability $p_w$.  In this case, it is more convenient to average over $T$ and work with
$B'(x) = \ex{T,U}{B(\Select(T,x,U))}$.
Then it turns out that $\widehat{B'}[s] = p_w^{|s|} \cdot \widehat{B}[s]$, where $|s|$ denotes the Hamming weight of the vector $s$.  Thus, it suffices to analyse the original program $B$ and show that for each $k\in \{1,\cdots,n\}$, the Fourier mass of $B$ restricted to $s$ of weight $k$ is at most $c_w^k$, where $c_w$ is a constant depending only on $w$ (not on $n$).  We prove that this is indeed the case for regular branching programs:

\begin{theorem} \label{thm:levelk-intro}
Let $B : \zo^n\rightarrow \zo$ be a function computed by an oblivious, read-once, regular branching program of width $w$.
Then for every $k\in \{1,\ldots,n\}$, we have
$$\sum_{s\in \zo^n:|s|=k} |\widehat{B}[s]| \leq (2w^2)^k.$$
\end{theorem}

Our proof of Theorem~\ref{thm:levelk-intro} relies on the main lemma of Braverman et al.~\cite{BRRY}, which
intuitively says that in a bounded-width, read-once, oblivious, regular branching program, only a constant number of bits have a significant effect on the acceptance probability.  More formally, if we sum, for every time step $i$ and all possible states $v$ at time $i$, the absolute difference between the acceptance probability after reading a 0 versus reading a 1 from state $v$, the total will be bounded by $\text{poly}(w)$ (independent of $n$). This directly implies a bound of $\text{poly}(w)$ on the Fourier mass of $B$ at the first level: the correlation of $B$ with a parity of weight $1$ is bounded by the effect of a single bit on the output of $B$. We then bound the correlation of $B$ with a parity of weight $k$ by the correlation of a \emph{prefix} of $B$ with a parity of weight $k-1$ times the effect of the remaining bit on $B$. Thus we inductively obtain the bound on the Fourier mass of $B$ at level $k$.

Our proof of Theorem \ref{thm:mainlemma-intro} for the case of a pseudorandom restriction $T$ uses the fact that we can decompose the high-order Fourier coefficients of an oblivious, read-once branching program $B'$ into products of low-order Fourier coefficients of ``subprograms'' (intervals of consecutive layers) of $B'$. 
Using an almost $O(\log n)$-wise independent choice of $T$ enables us to control the Fourier mass at level $O(\log n)$ for all subprograms of $B'$, which suffices to control the total Fourier mass of $B'$.

\subsection{Organization}

In Section \ref{SectionTechniques} we introduce the definitions and tools we use in our proof. In Section \ref{SubSectionBP} we formally define branching programs and explain our view of them as matrix-valued functions. In Sections \ref{SubSectionFourier} and \ref{SectionFourierBounds} we define the matrix-valued Fourier transform and explain how we use it.

Our results use Fourier analysis of regular branching programs to analyse pseudorandom generators. In Section \ref{SubSectionLowOrder}, we give a bound on the low-order Fourier coefficients of a read-once, oblivious, regular branching program (Theorem \ref{thm:levelk-intro}) using the main lemma of Braverman et al.~\cite{BRRY}. This yields a result about random restrictions, which we define and discuss in Section \ref{SubSectionCoin}. We extend the results about random restrictions to pseudorandom restrictions in Section \ref{SubSectionPRrestriction} and prove our main lemma (Theorem \ref{thm:mainlemma-intro}). Finally, in Section \ref{SubSectionPRG} we construct and analyse our pseudorandom generator, which proves the main result (Theorem \ref{thm:main-intro}).

In Section \ref{SectionGeneral} we show how to extend our techniques to general read-once, oblivious branching programs (Theorem \ref{thm:second-intro}). We conclude in Section \ref{SectionConclusion} by discussing directions for further work.

\section{Preliminaries} \label{SectionTechniques}

\subsection{Branching Programs} \label{SubSectionBP}

We define a length-$n$, width-$w$ \textbf{program} to be a function $B : \{0,1\}^n \times [w] \to [w]$, which takes a start state $u \in [w]$ and an input string $x \in \{0,1\}^n$ and outputs a final state $B[x](u)$.

Often we think of $B$ as having a fixed \textbf{start state} $u_0$ and a set of \textbf{accept states} $S \subset [w]$. Then $B$ \textbf{accepts} $x \in \{0,1\}^n$ if $B[x](u_0) \in S$. We say that $B$ \textbf{computes the function} $f : \{0,1\}^n \to \{0,1\}$ if $f(x)=1$ if and only if $B[x](u_0) \in S$.

In our applications, the input $x$ is randomly (or pseudorandomly) chosen, in which case a program can be viewed as a Markov chain randomly taking initial states to final states. For each $x \in \{0,1\}^n$, we let $B[x] \in \{0,1\}^{w \times w}$ be a matrix defined by $$B[x](u,v) = 1 \iff B[x](u)=v.$$ 

For a random variable $X$ on $\{0,1\}^n$, we have $\ex{X}{B[X]} \in [0,1]^{w \times w},$ where $\ex{R}{f(R)}$ is the \textbf{expectation} of a function $f$ with respect to a random variable $R$. Then the entry in the $u^\text{th}$ row and $v^\text{th}$ column $\ex{X}{B[X]}(u,v)$ is the probability that $B$ takes the initial state $u$ to the final state $v$ when given a random input from the distribution $X$---that is, $$\ex{X}{B[X]}(u,v) = \pr{X}{B[X](u)=v},$$ where $\pr{R}{e(R)}$ is the \textbf{probability} of an event $e$ with respect to the random variable $R$.

A branching program reads one bit of the input at a time (rather than reading $x$ all at once) maintaining only a state in $[w] = \{1,2, \cdots, w\}$ at each step. We capture this restriction by demanding that the program be composed of several smaller programs, as follows.

Let $B$ and $B'$ be width-$w$ programs of length $n$ and $n'$ respectively. We define the \textbf{concatenation} $B \circ B' : \{0,1\}^{n+n'} \times [w] \to [w]$ of $B$ and $B'$ by $$(B \circ B')[x \circ x'](u) := B'[x'](B[x](u)),$$ which is a width-$w$, length-$(n+n')$ program. That is, we run $B$ and $B'$ on separate inputs, but the final state of $B$ becomes the start state of $B'$. Concatenation corresponds to matrix multiplication---that is, $(B \circ B')[x \circ x'] = B[x] \cdot B'[x']$, where the two programs are concatenated on the left hand side and the two matrices are multiplied on the right hand side. 

A length-$n$, width-$w$, \textbf{ordered branching program} is a program $B$ that can be written $B = B_1 \circ B_2 \circ \cdots \circ B_n$, where each $B_i$ is a length-$1$ width-$w$ program. We refer to $B_i$ as the $i^\text{th}$ \textbf{layer} of $B$. We denote the \textbf{subprogram} of $B$ from layer $i$ to layer $j$ by $B_{ i \cdots j} := B_i \circ B_{i+1} \circ \cdots \circ B_j$.

General read-once, oblivious branching programs (a.k.a. unordered branching programs) can be reduced to the ordered case by a permutation of the input bits. Formally, a \textbf{read-once, oblivious branching program} $B$ is an ordered branching program $B'$ composed with a permutation $\pi$. That is, $B[x]=B'[\pi(x)]$, where the $i^\text{th}$ bit of $\pi(x)$ is the $\pi(i)^\text{th}$ bit of $x$

For a program $B$ and an arbitrary distribution $X$, the matrix $\ex{X}{B[X]}$ is \textbf{stochastic}---that is, $$\sum_v \ex{X}{B[X]}(u,v) = 1$$ for all $u$ and $\ex{X}{B[X]}(u,v) \geq 0$ for all $u$ and $v$. A program $B$ is called a \textbf{regular program} if the matrix $\ex{U}{B[U]}$ is \textbf{doubly stochastic}---that is, both $\ex{U}{B[U]}$ and its transpose $\ex{U}{B[U]}^*$ are stochastic. A program $B$ is called a \textbf{permutation program} if $B[x]$ is a permutation matrix for every $x$ or, equivalently, $B[x]$ is doubly stochastic. Note that a permutation program is necessarily a regular program and, if both $B$ and $B'$ are regular or permutation programs, then so is their concatenation.

A regular program $B$ has the property that the uniform distribution is a stationary distribution of the Markov chain $\ex{U}{B[U]}$, whereas, if $B$ is a permutation program, the uniform distribution is stationary for $\ex{X}{B[X]}$ for \emph{any} distribution $X$.

A \textbf{regular branching program} is a branching program where each layer $B_i$ is a regular program and likewise for a \textbf{permutation branching program}.

\subsection{Norms}

We are interested in constructing a random variable $X$ (the output of the pseudorandom generator) such that $\ex{X}{B[X]} \approx \ex{U}{B[U]}$, where $U$ is uniform on $\{0,1\}^n$. Throughout we use $U$ to denote the \textbf{uniform distribution}. The error of the pseudorandom generator will be measured by the norm of the matrix $\ex{X}{B[X]} - \ex{U}{B[U]}$.

For a matrix $A \in \mathbb{R}^{w \times w}$, define the \textbf{$\rho$ operator norm} of $A$ by $$\norm{A}_\rho = \max_x \frac{\norm{xA}_\rho}{\norm{x}_\rho},$$ where $\rho$ specifies a vector norm (usually $1$, $2$, or $\infty$ norm). Define the \textbf{Frobenius norm} of $A \in \mathbb{R}^{w \times w}$ by $$\frob{A}^2 = \sum_{u,v} A(u,v)^2 = \text{trace}(A^*A) = \sum_\lambda |\lambda|^2,$$ where $A^*$ is the (conjugate) transpose of $A$ and the last sum is over the singular values $\lambda$ of $A$. Note that $\norm{A}_2 \leq \frob{A}$ for all $A$.

We almost exclusively use the Euclidean norm ($\norm{x}_2 = \sqrt{\sum_i x(i)^2}$) and the corresponding spectral norm ($\norm{A}_2 = \max_{\lambda} |\lambda|$). This is not crucial; our results would work with any reasonable norm.

\subsection{Fourier Analysis} \label{SubSectionFourier}

Let $B : \{0,1\}^n \to \mathbb{R}^{w\times w}$ be a matrix-valued function (such as given by a length-$n$, width-$w$ branching program). Then we define the \textbf{Fourier transform} of $B$ as a matrix-valued function $\widehat{B} : \{0,1\}^n \to \mathbb{R}^{w \times w}$ given by $$\widehat{B}[s] := \ex{U}{B[U] \chi_s(U)},$$ where $s \in \{0,1\}^n$ (or, equivalently, $s \subset [n]$) and $$\chi_s(x) = (-1)^{\sum_i x(i) \cdot s(i)} = \prod_{i \in s} (-1)^{x(i)}.$$ We refer to $\widehat{B}[s]$ as the $s^\text{th}$ \textbf{Fourier coefficient} of $B$. The \textbf{order} of a Fourier coefficient $\widehat{B}[s]$ is $|s|$---the \textbf{Hamming weight} of $s$, which is the size of the set $s$ or the number of $1$s in the string $s$. Note that this is equivalent to taking the real-valued Fourier transform of each of the $w^2$ entries of $B$ separately, but we will see below that this matrix-valued Fourier transform is nicely compatible with matrix algebra.

For a random variable $X$ over $\{0,1\}^n$ we define its $s^\text{th}$ \textbf{Fourier coefficient} as $$\widehat{X}(s) := \ex{X}{\chi_s(X)},$$ which, up to scaling, is the same as taking the real-valued Fourier transform of the probability mass function of $X$.
We have the following useful properties.
\begin{lemma} \label{LemmaFourier}
Let $A, B: \{0,1\}^n \to \mathbb{R}^{w\times w}$ be matrix valued functions. Let $X$, $Y$, and $U$ be independent random variables over $\{0,1\}^n$, where $U$ is uniform. Let $s,t \in \{0,1\}^n$. Then we have the following.
\begin{itemize}
\item Decomposition: If $C[x \circ y] = A[x] \cdot B[y]$ for all $x,y \in \{0,1\}^n$, then $\widehat{C}[s \circ t] = \widehat{A}[s] \cdot \widehat{B}[t]$.
\item Expectation: $\ex{X}{B[X]} = \sum_s \widehat{B}[s] \widehat{X}(s)$.
\item Fourier Inversion for Matrices: $B[x] = \sum_s \widehat{B}[s] \chi_s(x)$.
\item Fourier Inversion for Distributions: $\pr{X}{X=x} = \ex{U}{\widehat{X}(U) \chi_U(x)}$.
\item Convolution for Distributions: If $Z = X \oplus Y$, then $\widehat{Z}(s) = \widehat{X}(s) \cdot \widehat{Y}(s)$.
\item Parseval's Identity: $\sum_{s \in \{0,1\}^n} \frob{\widehat{B}[s]}^2 = \ex{U}{\frob{B[U]}^2}$.
\item Convolution for Matrices: If, for all $x \in \{0,1\}^n$, $C[x] = \ex{U}{A[U] \cdot B[U \oplus x]}$, then $\widehat{C}[s] = \widehat{A}[s] \cdot \widehat{B}[s]$.
\end{itemize}
\end{lemma}

The Decomposition property is what makes the matrix-valued Fourier transform more convenient than separately taking the Fourier transform of the matrix entries as done in \cite{BPW}. If $B$ is a length-$n$ width-$w$ branching program, then, for all $s \in \{0,1\}^n$, $$\widehat{B}[s] = \widehat{B}_1[s_1] \cdot \widehat{B}_2[s_2] \cdot \cdots \cdot \widehat{B}_n[s_n].$$

\subsection{Small-Bias Distributions} \label{SubSectionSmallBias}

The \textbf{bias} of a random variable $X$ over $\{0,1\}^n$ is defined as $$\text{bias}(X) := \max_{s \ne 0} |\widehat{X}(s)|.$$ A distribution is \textbf{$\varepsilon$-biased} if it has bias at most $\varepsilon$. Note that a distribution has bias $0$ if and only if it is uniform. Thus a distribution with small bias is an approximation to the uniform distribution. We can sample an $\varepsilon$-biased distribution $X$ on $\{0,1\}^n$ with seed length $O(\log(n/\varepsilon))$ and using space $O(\log(n/\varepsilon))$ \cite{NaorNa93,AGHP}.

Small-bias distributions are useful pseudorandom generators: A $\varepsilon$-biased random variable $X$ is indistinguishable from uniform by any linear function (a parity of a subset of the bits of $X$). That is, for any $s \subset [n]$, we have $\left| \ex{X}{\bigoplus_{i \in s} X_i} - 1/2\right|\leq 2\varepsilon$. Small bias distributions are known to be good pseudorandom generators for width-$2$ branching programs \cite{BogdanovDvVeYe09}, but not width-$3$. For example, the uniform distribution over $\{x \in \{0,1\}^n : |x| \mod{3} = 0 \}$ has bias $2^{-\Omega(n)}$, but does not fool width-$3$, ordered, permutation branching programs.

\subsection{Fourier Mass} \label{SectionFourierBounds}

We analyse small bias distributions as pseudorandom generators for branching programs using Fourier analysis. Intuitively, the Fourier transform of a branching program expresses that program as a linear combination of linear functions (parities), which can then be fooled using a small-bias space.

Define the \textbf{Fourier mass} of a matrix-valued function $B$ to be $$L_\rho(B) := \sum_{s \ne 0} \norm{\widehat{B}[s]}_\rho.$$ Also, define the \textbf{Fourier mass of $B$ at level $k$} as $$L_\rho^k(B) := \sum_{s \in \{0,1\}^n : |s|=k} \norm{\widehat{B}[s]}_\rho.$$ Note that $L_\rho(B) =\sum_{k \geq 1} L_\rho^k(B)$.

The Fourier mass is unaffected by order:

\begin{lemma} \label{LemmaFourierPermutation}
Let $B, B' : \{0,1\}^n \to \mathbb{R}^{w \times w}$ be matrix-valued functions satisfying $B[x]=B'[\pi(x)]$, where $\pi : [n] \to [n]$ is a permutation. Then, for all $s \in \{0,1\}^n$, $\widehat{B}[s]=\widehat{B'}[\pi(s)]$. In particular, $L_\rho(B)=L_\rho(B')$ and $L_\rho^k(B)=L_\rho^k(B')$ for all $k$ and $\rho$.
\end{lemma}

Lemma \ref{LemmaFourierPermutation} implies that the Fourier mass of any read-once, oblivious branching program is equal to the Fourier mass of the corresponding ordered branching program.

If $L_\rho(B)$ is small, then $B$ is fooled by a small-bias distribution:

\begin{lemma} \label{LemmaBiasMass}
Let $B$ be a length-$n$, width-$w$ branching program. Let $X$ be a $\varepsilon$-biased random variable on $\{0,1\}^n$. For any matrix norm $\norm{\cdot}_\rho$, we have $$\norm{\ex{X}{B[X]}-\ex{U}{B[U]}}_\rho = \norm{\sum_{s \ne 0} \widehat{B}[s] \widehat{X}(s)}_\rho \leq L_\rho(B) \varepsilon.$$
\end{lemma}

In the worst case $L_2(B) = 2^{\Theta(n)}$, even for a length-$n$ width-$3$ permutation branching program $B$. For example, the program $B_{\text{mod 3}}$ that computes the Hamming weight of its input modulo $3$ has exponential Fourier mass.

We show that, using `restrictions', we can ensure that $L_\rho(B)$ is small.

\section{Fourier Analysis of Regular Branching Programs} \label{SubSectionLowOrder}

We use a result by Braverman et al.~\cite{BRRY}. The following is a Fourier-analytic reformulation of their result.

\begin{lemma}[{\cite[Lemma 4]{BRRY}}] \label{LemmaBRRY}
Let $B$ be a length-$n$, width-$w$, ordered, regular branching program. Then $$\sum_{1 \leq i \leq n} \norm{ \widehat{B_{i \cdots n}}[1 \circ 0^{n-i}] }_2 \leq 2w^2.$$
\end{lemma}
Braverman et al.~instead consider the sum, over all $i \in [n]$ and all states $u\in [w]$ at layer $i$, of the difference in acceptance probabilities if we run the program starting at $v$ with a $0$ followed by random bits versus a $1$ followed by random bits. They refer to this quantity as the \textbf{weight} of $B$. Their result can be expressed in Fourier-analytic terms by considering subprograms $B_{i \cdots n}$ that are the original program with the first $i-1$ layers removed: $$\sum_{1 \leq i \leq n} \norm{ \widehat{B_{i \cdots n}}[1 \circ 0^{n-i}] q}_1 \leq 2(w-1)$$ for any $q \in \{0,1\}^w$ with $\sum_u q(u)=1$. (The vector $q$ can be used to specify the accept states of $B$, and the $v^\text{th}$ row of $ \widehat{B_{i \cdots n}}[1 \circ 0^{n-i}] q$ is precisely the difference in acceptance probabilities mentioned above.) By summing over all $w$ possible $q$, we obtain $$\sum_{i \in [n]} \sum_u \norm{\widehat{B_{i \cdots n}}[1 \circ 0^{n-i}](\cdot,u)}_1 \leq 2w(w-1).$$  This implies Lemma \ref{LemmaBRRY}, as the spectral norm of a matrix is bounded by the sum of the $1$-norms of the columns. For completeness, we include a direct proof of Lemma \ref{LemmaBRRY} in Appendix \ref{AppendixBRRY}.

Lemma \ref{LemmaBRRY} is similar (but not identical) to a bound on the first-order Fourier coefficients of a regular branching program:  The term $\widehat{B_{i \cdots n}}[1 \circ 0^{n-i}]$ measures the effect of the $i^\text{th}$ bit on the output of $B$ when we start the program at layer $i$, whereas the $i^\text{th}$ first-order Fourier coefficient $\widehat{B}[0^{i-1} \circ 1 \circ 0^{n-i}]$ measures the effect of the $i^\text{th}$ bit when we start at the first layer and run the first $i-1$ layers with random bits. This difference allows us to use Lemma \ref{LemmaBRRY} to obtain a bound on all low-order Fourier coefficients of a regular branching program:

\begin{theorem} \label{TheoremLow}
Let $B$ be a length-$n$, width-$w$, read-once, oblivious, regular branching program. Then $$L_2^k(B) := \sum_{s \in \{0,1\}^n : |s|=k} \norm{\widehat{B}[s]}_2 \leq (2w^2)^{k}.$$
\end{theorem}

The key point is that the bound does not depend on $n$, even though we are summing ${ n \choose k }$ terms. 

\begin{proof} 
By Lemma \ref{LemmaFourierPermutation}, we may assume that $B$ is ordered. We perform an induction on $k$. If $k=0$, then there is only one Fourier coefficient to bound---namely, $\widehat{B}[0^n] = \ex{U}{B[U]}$. Since $\ex{U}{B[U]}$ is doubly stochastic, the base case follows from the fact that every doubly stochastic matrix has spectral norm $1$. Now suppose the bound holds for $k$ and consider $k+1$. We split the Fourier coefficients based on where the last $1$ is:
\begin{align*}
\lefteqn{\sum_{s \in \{0,1\}^n : |s|=k+1} \norm{\widehat{B}[s]}_2}~~~~~~~~~~&\\ =& \sum_{1 \leq i \leq n} \sum_{s \in \{0,1\}^{i-1} : |s|=k} \norm{\widehat{B}[s \circ 1 \circ 0^{n-i}]}_2\\
=& \sum_{1 \leq i \leq n} \sum_{s \in \{0,1\}^{i-1} : |s|=k} \norm{\widehat{B_{1 \cdots i-1}}[s] \cdot \widehat{B_{i \cdots n}}[1 \circ 0^{n-i}]}_2~~~\text{(by Lemma \ref{LemmaFourier} (Decomposition))}\\
\leq& \sum_{1 \leq i \leq n} \sum_{s \in \{0,1\}^{i-1} : |s|=k} \norm{\widehat{B_{1 \cdots i-1}}[s]}_2 \cdot \norm{ \widehat{B_{i \cdots n}}[1 \circ 0^{n-i}]}_2\\
\leq& \sum_{1 \leq i \leq n} (2w^2)^{k} \cdot \norm{ \widehat{B_{i \cdots n}}[1 \circ 0^{n-i}]}_2~~~\text{(by the induction hypothesis)}\\
\leq& (2w^2)^{k} \cdot 2 w^2~~~\text{(by Lemma \ref{LemmaBRRY})}\\
=&(2w^2)^{k+1},\\
\end{align*}
as required.
\end{proof}

\section{Random Restrictions} \label{SubSectionCoin}

Our results involve restricting branching programs. However, our use of restrictions is different from elsewhere in the literature. Here, as in \cite{GMRTV}, we use (pseudorandom) restrictions in the usual way, but we \emph{analyse} them by averaging over the \emph{unrestricted} bits. Formally, we define a restriction as follows.

\begin{definition}
For $t \in \{0,1\}^n$ and a length-$n$ branching program $B$, let $B|_t$ be the \textbf{restriction} of $B$ to $t$---that is, $B|_t : \{0,1\}^n \to \mathbb{R}^{w \times w}$ is a matrix-valued function given by $B|_t [x] := \ex{U}{B[\mathrm{Select}(t,x,U)]}$, where $U$ is uniform on $\{0,1\}^n$. 
\end{definition}
Here $\text{Select}$ takes a set $t \subset [n]$, a string $x \in \{0,1\}^n$, and a string $y$ of length at least $n-|t|$ and produces a string of length $n$ given by $$\text{Select}(t,x,y)(i) = \left\{ \begin{array}{cl} x(i) & i \in t\\ y(|[i]\backslash t|) & i \in [n]\backslash t  \end{array} \right\}.$$ 
Intuitively, $\text{Select}(t,x,y)$ `stretches' $y$ by `skipping' the bits in $t$ and using bits from $x$ instead. For example, $\mathrm{Select}(0101000,1111111,00001)=0101001$.

The most important aspect of restrictions is how they relate to the Fourier transform: For all $B$, $s$, and $t$, we have $\widehat{B|_t}[s] = \widehat{B}[s]$ if $s \subset t$ and $\widehat{B|_t}[s] = 0$ otherwise. The restriction $t$ `kills' all the Fourier coefficients that are not contained in it. This means that a restriction significantly reduces the Fourier mass: 

\begin{lemma} \label{LemmaRestrictionMass}
Let $B$ be a length-$n$, width-$w$ program. Let $T$ be $n$ independent random bits each with expectation $p$. Then $$\ex{T}{L_2(B|_T)} = \sum_{s \ne 0} p^{|s|} \norm{\widehat{B}[s]}_2.$$
\end{lemma}
\begin{proof}
$$\ex{T}{L_2(B|_T)} = \ex{T}{\sum_{s \ne 0 : s \subset T} \norm{\widehat{B}[s]}_2} = \sum_{s \ne 0} \pr{T}{s \subset T} \norm{\widehat{B}[s]}_2 = \sum_{s \ne 0} p^{|s|} \norm{\widehat{B}[s]}_2.$$
\end{proof}

We will overload notation as follows. Let $B = B' \circ B''$ be a branching program, where $B'$ has length $n'$ and $B''$ has length $n''$ and $B$ has length $n=n'+n''$. For $t \in \{0,1\}^n$, we define $B'|_t = B'|_{t'}$ and $B''|_t = B'|_{t''}$ where $t' \in \{0,1\}^{n'}$, $t'' \in \{0,1\}^{n''}$ and $t = t' \circ t''$. Then $B|_t = B'|_{t'} \circ B''|_{t''} = B'|_t \circ B''|_t$.

Now we use Theorem \ref{TheoremLow} to prove a result about random restrictions of regular branching programs:

\begin{proposition} \label{PropositionRandomRestriction}
Let $B$ be a length-$n$, width-$w$, read-once, oblivious, regular branching program. Let $T$ be $n$ independent random bits each with expectation $p < 1/2w^2$. Then $$\ex{T}{L_2(B|_T)} \leq \frac{2w^2 \cdot p}{1-2 w^2 \cdot p}.$$
\end{proposition}
In particular, if $p \leq 1/4w^2$, then $\ex{T}{L_2(B|_T)} \leq 1$.
\begin{proof}
By Lemma \ref{LemmaRestrictionMass}, we have,
\begin{align*}
\ex{T}{L_2(B|_T)} =& \sum_{s \ne 0} p^{|s|} \cdot \norm{\widehat{B}[s]}_2\\
=& \sum_{1 \leq k \leq n} p^k \cdot \sum_{s : |s|=k} \norm{\widehat{B}[s]}_2\\
\leq& \sum_{1 \leq k \leq n} p^k \cdot (2w^2)^{k}~~~~~\text{(by Theorem \ref{TheoremLow})}\\
\leq& \sum_{k \geq 1} (2w^2 \cdot p)^{k}\\
=& \frac{2 w^2 \cdot p}{1-2 w^2 \cdot p}.\\
\end{align*}
\end{proof}

\paragraph{Relation to Coin Theorem}

The Coin Theorem of Brody and Verbin~\cite{BrodyVerbin} shows that general (non-regular) oblivious, read-once branching programs of width $w$ cannot distinguish $n$ independent and unbiased coin flips from ones with bias $1/(\log n)^{\Theta(w)}$, and they show that this bound is the best possible. Braverman et al.~\cite{BRRY} show that regular, oblivious, read-once branching programs of width $w$ cannot distinguish coins with bias $\Theta(1/w)$ from unbiased ones. (They state this result in terms of $\alpha$-biased spaces.)

If $Z$ is $n$ independent coin flips with bias $p$ and $B$ is a branching program, then $$\norm{ \ex{Z}{B[Z]} - \ex{U}{B[U]} } = \norm{ \sum_{s \ne 0} \widehat{B}[s] \widehat{Z}(s) } = \norm{ \sum_{s \ne 0} p^{|s|} \widehat{B}[s] }.$$
Thus Proposition \ref{PropositionRandomRestriction} implies a Coin Theorem showing that read-once, oblivious, regular branching programs cannot distinguish coins with bias $p$ for some $p=\Theta(1/w^2)$ from unbiased ones. This Coin Theorem is weaker than the Braverman et al.~result, which gives $p=\Theta(1/w)$. However, Proposition \ref{PropositionRandomRestriction} gives more than a Coin Theorem, as the sum is taken outside the norm---that is, we bound $$\sum_{s \ne 0} p^{|s|} \norm{\widehat{B}[s]}_2.$$ This distinction is important for our purposes, as it will allow us to reason about small-bias distributions and restrictions together.

\section{Pseudorandom Restrictions} \label{SubSectionPRrestriction}

To analyse our generator, we need a pseudorandom version of Proposition \ref{PropositionRandomRestriction}. That is, we need to prove that, for a \emph{pseudorandom} $T$ (generated using few random bits), $L_2(B|_T)$ is small. We will generate $T$ using an almost $O(\log n)$-wise independent distribution:

\begin{definition} \label{DefinitionLimitedIndependence}
A random variable $X$ on $\Omega^n$ is \textbf{$\delta$-almost $k$-wise independent} if, for any\\ $I=\{i_1, i_2, \cdots, i_k\} \subset [n]$ with $|I| = k$, the coordinates $(X_{i_1}, X_{i_2}, \cdots, X_{i_k}) \in \Omega^k$ are $\delta$ statistically close to being independent---that is, for all $T \subset \Omega^k$, $$\left|\sum_{x \in T} \left( \pr{X}{(X_{i_1},X_{i_2},\cdots,X_{i_k}) = x} - \prod_{l \in [k]} \pr{X}{X_{i_l} = x_l} \right) \right| \leq \delta.$$ We say that $X$ is \textbf{$k$-wise independent} if it is $0$-almost $k$-wise independent.
\end{definition}

We can sample a random variable $X$ on $\{0,1\}^n$ that is $\delta$-almost $k$-wise independent such that each bit has expectation $p=2^{-d}$ using $O(kd+\log(1/\delta)+d\log(nd))$ random bits. See Lemma \ref{LemmaLimitedIndependenceSeed} for more details.

Our main lemma (stated informally as Theorem \ref{thm:mainlemma-intro}) is as follows.

\begin{theorem}[Main Lemma] \label{TheoremMainLemma}
Let $B$ be a length-$n$, width-$w$, read-once, oblivious, regular branching program. Let $T$ be a random variable over $\{0,1\}^n$ where each bit has expectation $p$ and the bits are $\delta$-almost $2k$-wise independent. Suppose $p \leq (2w)^{-2}$ and $\delta \leq (2w)^{-4k}$. Then $$\pr{T}{L_2(B|_T) \leq (2w^2)^k } \geq 1 -  n^4 \cdot \frac{2}{2^k}.$$
\end{theorem}

In particular, we show that, for $w=O(1)$, $k=O(\log n)$, and $\delta=1/\text{poly}(n)$, we have $L_2(B|_T) \leq \text{poly}(n)$ with probability $1-1/\text{poly}(n)$.

First we show that the Fourier mass at level $O(\log n)$ is bounded by $1/n$ with high probability. This also applies to all subprograms---that is, $$\pr{T}{\forall i,j ~~L_2^k(B_{i \cdots j}|_T) \leq 1/n} \geq 1-1/{\text{poly}(n)}.$$

\begin{lemma} \label{LemmaKbound}
Let $B$ be a length-$n$, width-$w$, ordered, regular branching program. Let $T$ be a random variable over $\{0,1\}^n$ where each bit has expectation $p$ and the bits are $\delta$-almost $k$-wise independent. If $p \leq (2w)^{-2}$ and $\delta \leq (2w)^{-2k}$, then, for all $\beta > 0$, $$\pr{T}{\forall 1 \leq i \leq j \leq n ~~ L_2^k(B_{i \cdots j}|_T) \leq \beta} \geq  1 - n^2\frac{2}{2^k \beta}.$$
\end{lemma}
\begin{proof}
By Theorem \ref{TheoremLow}, for all $i$ and $j$, $$\ex{T}{L_2^k(B_{i \cdots j}|_T)} = \sum_{s \subset \{i \cdots j\} : |s|=k} \pr{T}{s \subset T} \norm{\widehat{B_{i \cdots j}}[s]}_2 \leq (2w^2)^k (p^k + \delta) \leq \frac{2}{2^k}.$$ The result now follows from Markov's inequality and a union bound.
\end{proof}

Now we use Lemma \ref{LemmaKbound} to bound the Fourier mass at higher levels. We decompose high-order ($k'\geq2k$) Fourier coefficients into low-order ($k \leq k' < 2k$) ones:

\begin{lemma} \label{LemmaWellOrder}
Let $B$ be a length-$n$, ordered branching program and $t \in \{0,1\}^n$. Suppose that, for all $i$, $j$, and $k'$ with $1 \leq i \leq j \leq n$ and $k \leq k' < 2k$, $L_2^{k'}(B_{i \cdots j}|_t) \leq 1/n$. Then, for all $k'' \geq k$ and all $i$ and $j$, $L_2^{k''}(B_{i \cdots j}|_t) \leq 1/n$.
\end{lemma}
\begin{proof}
Suppose otherwise and let $k''$ be the smallest $k''\geq k$ such that $L_2^{k''}(B_{i \cdots j}|_t) > 1/n$ for some $i$ and $j$. Clearly $k'' \geq 2k$. So, by minimality, the result holds for $k''-k$. Fix $i$ and $j$. Now
\begin{align*}
L_2^{k''}(B_{i \cdots j}|_t) =& \sum_{s \in \{0,1\}^{j-i+1} : |s|=k''} \norm{\widehat{B_{i \cdots j}}[s]}_2\\
\leq& \sum_{l \in \{i \cdots j\}} \sum_{s \in \{0,1\}^{l-i+1} : |s|=k} \sum_{s' \in \{0,1\}^{j-l} : |s'|=k''-k} \norm{\widehat{B_{i \cdots l}}[s] \cdot \widehat{B_{l+1 \cdots j}}[s']}_2\\
\leq& \sum_{l \in \{i \cdots j\}} \left(\sum_{s \in \{0,1\}^{l-i+1} : |s|=k} \norm{\widehat{B_{i \cdots l}}[s]}_2 \right) \left( \sum_{s' \in \{0,1\}^{j-l} : |s'|=k''-k} \norm{\widehat{B_{l+1 \cdots j}}[s']}_2 \right)\\
\leq& \sum_{l \in \{i \cdots j\}} \frac{1}{n^2}\\
\leq& \frac{1}{n}.\\
\end{align*}
Since $i$ and $j$ were arbitrary, this contradicts our supposition and proves the result.
\end{proof}

\begin{proof}[Proof of Theorem \ref{TheoremMainLemma}]
By Lemma \ref{LemmaFourierPermutation}, we may assume that $B$ is ordered. By Lemma \ref{LemmaKbound} and a union bound, $$\pr{T}{\forall k \leq k' < 2k ~ \forall 1 \leq i \leq j \leq n ~~ L_2^{k'}(B_{i \cdots j}|_T) \leq \frac{1}{n}} \geq  1 - n^4 \cdot \frac{2}{2^k}.$$ Lemma \ref{LemmaWellOrder} thus implies that $$\pr{T}{\forall k'' \geq k ~ \forall 1 \leq i \leq j \leq n ~~ L_2^{k''}(B_{i \cdots j}|_T) \leq \frac{1}{n}} \geq  1 - n^4 \cdot \frac{2}{2^k}.$$ Thus $$\pr{T}{\sum_{k'' \geq k} L_2^{k''}(B|_T) \leq 1} \geq 1 - n^4 \cdot \frac{2}{2^k}.$$ By Theorem \ref{TheoremLow}, $$\sum_{0 < k' < k} L_2^{k'}(B|_T)  \leq \sum_{0< k' < k} L_2^{k'}(B)  \leq \sum_{0 \leq k' \leq k-1} (2w^2)^{k'} = \frac{(2w^2)^k-1}{2w^2-1} \leq (2w^2)^k-1.$$ The result now follows.
\end{proof}

\section{The Pseudorandom Generator} \label{SubSectionPRG}

Our main result Theorem \ref{thm:main-intro} is stated more formally as follows.

\begin{theorem}[Main Result] \label{TheoremPRG}
There exists a pseudorandom generator family $G_{n, w, \varepsilon} : \{0,1\}^{s_{n, w, \varepsilon}} \to \{0,1\}^n$ with seed length $$s_{n, w, \varepsilon} = O(w^2 \log(w) \log(n) \log(nw/\varepsilon) + w^4 \log^2(w/\varepsilon))$$ such that, for any length-$n$, width-$w$, read-once, oblivious (but unordered), permutation branching program $B$ and $\varepsilon > 0$, $$\norm{\ex{U_{s_{n, w, \varepsilon}}}{B[G_{n,w,\varepsilon}(U_{s_{n, w, \varepsilon}})]}-\ex{U}{B[U]}}_2 \leq \varepsilon.$$ Moreover, $G_{n, w, \varepsilon}$ can be computed in space $O(s_{n,w,\varepsilon})$.
\end{theorem}

The following lemma gives the basis of our pseudorandom generator.

\begin{lemma} \label{LemmaOneStep}
Let $B$ be a length-$n$, width-$w$, read-once, oblivious, regular branching program. Let $\varepsilon \in (0,1)$. Let $T$ be a random variable over $\{0,1\}^n$ that is $\delta$-almost $2k$-wise independent and each bit has expectation $p$, where we require $$p \leq 1/4w^2 ,~~~~ k \geq \log_2 \left( 4 \sqrt{w} n^4 / \varepsilon \right), ~~~~\text{and}~~~~\delta \leq (2w)^{-4k}.$$ Let $U$ be uniform over $\{0,1\}^n$. Let $X$ be a $\mu$-biased random variable over $\{0,1\}^n$ with $\mu \leq \varepsilon (2w^2)^{-k}$. Then $$\norm{\ex{T,X,U}{B[\mathrm{Select}(T,X,U)]} - \ex{U}{B[U]}}_2 \leq 2\varepsilon.$$
\end{lemma}

Theorem \ref{TheoremMainLemma} says that with high probability over $T$, $B|_T = \ex{U}{B[\mathrm{Select}(T,\cdot,U)]}$ has small Fourier mass. This implies that $B|_T$ is fooled by small bias $X$ and thus $$\ex{T,X,U}{B[\mathrm{Select}(T,X,U)]} \approx \ex{T,U,U'}{B[\mathrm{Select}(T,U',U)]} = \ex{U}{B[U]}.$$

\begin{proof}
For $t \in \{0,1\}^n$, we have
\begin{align*}
\norm{\ex{X,U}{B[\text{Select}(t,X,U)]} - \ex{U}{B[U]}}_2 =& {\norm{\ex{X}{B|_t[X]} - \ex{U}{B[U]}}_2}\\
=& {\norm{\sum_{s \ne 0} \widehat{B|_t}[s] \widehat{X}(s)}_2}\\
\leq& {\sum_{s \ne 0} \norm{\widehat{B|_t}[s]}_2 |\widehat{X}(s)|}\\
\leq& {L_2(B|_t) \mu}.\\
\end{align*}
We apply Theorem \ref{TheoremMainLemma} and, with probability at least $1- 2 \cdot n^4/2^{k}$ over $T$, we have $L_2(B|_T) \mu \leq \varepsilon$. Thus
\begin{align*}
\norm{\ex{T,X,U}{B[\text{Select}(T,X,U)]} - \ex{U}{B[U]}}_2 \leq& \pr{T}{L_2(B|_T) \mu > \varepsilon} \max_{t} \norm{\ex{X,U}{B[\text{Select}(t,X,U)]} - \ex{U}{B[U]}}_2 \\&+ \pr{T}{L_2(B|_T) \mu \leq \varepsilon} \varepsilon\\
\leq& 2 n^4 / 2^{k} \cdot 2\sqrt{w} + \varepsilon\\
\leq& 2 \varepsilon.\\
\end{align*}
\end{proof}

Now we use the above results to construct our pseudorandom generator for a read-once, oblivious, permutation branching program $B$.

Lemma \ref{LemmaOneStep} says that, if we define $\overline{B}_{t,x}[y]:=B[\mathrm{Select}(t,x,y)]$, then $\ex{T,X,U}{\overline{B}_{T,X}[U]} \approx \ex{U}{B[U]}$, where $T$ is almost $k$-wise independent with each bit having expectation $p$ and $X$ has small bias. So now we need only construct a pseudorandom generator for $\overline{B}_{t,x}$, which is a length-$(n-|t|)$ permutation branching program. Then $$\ex{T,X,\tilde{U}}{\overline{B}_{T,X}[\tilde{U}]}  \approx \ex{T,X,U}{\overline{B}_{T,X}[U]} \approx \ex{U}{B[U]},$$ where $\tilde{U}$ is the output of the pseudorandom generator for $\overline{B}_{t,x}$. We construct $\tilde{U} \in \{0,1\}^{n-|T|}$ recursively; each time we recurse, the required output length is reduced to $n-|T| \approx n(1-p)$. Thus after $O(\log(n)/p)$ levels of recursion the required output length is constant.

The only place where the analysis breaks down for regular branching programs is when we recurse. If $B$ is only a \emph{regular} branching program, $\overline{B}_{t,x}$ may not be regular. However, if $B$ is a \emph{permutation} branching program, then $\overline{B}_{t,x}$ is too. Essentially, the only obstacle to generalising the analysis to regular branching programs is that regular branching programs are not closed under restrictions.

The pseudorandom generator is formally defined as follows.
\begin{quote}
\begin{center}\textbf{Algorithm for $G_{n, w, \varepsilon} : \{0,1\}^{s_{n, w, \varepsilon}} \to \{0,1\}^n$.}\end{center}
\begin{itemize}
\item[Parameters:] $n \in \mathbb{N}$, $w \in \mathbb{N}$, $\varepsilon>0$.
\item[Input:] A random seed of length $s_{n, w, \varepsilon}$.
\item[1.] Compute appropriate values of $p \in [1/8w^2, 1/4w^2]$, $k \geq \log_2 \left( 4 \sqrt{w} n^4 / \varepsilon \right)$, $\delta=\varepsilon (2w)^{-4k}$, and $\mu = \varepsilon (2w^2)^{-k}$.\footnote{For the purposes of the analysis we assume that $p$, $k$, $\delta$, and $\mu$ are the same at every level of recursion. So if $G_{n,w,\varepsilon}$ is being called recursively, use the same values of $p$, $k$, $\delta$, and $\mu$ as at the previous level of recursion.}
\item[2.] If $n \leq (4 \cdot \log_2(2/\varepsilon)/p)^2$, output $n$ truly random bits and stop.
\item[3.] Sample $T \in \{0,1\}^n$ where each bit has expectation $p$ and the bits are $\delta$-almost $2k$-wise independent.
\item[4.] If $|T|<pn/2$, output $0^n$ and stop.
\item[5.] Recursively sample $\tilde{U} \in \{0,1\}^{\lfloor n(1-p/2) \rfloor}$. i.e. $\tilde{U}=G_{\lfloor n(1-p/2) \rfloor,w,\varepsilon}(U)$. 
\item[6.] Sample $X \in \{0,1\}^n$ from a $\mu$-biased distribution.
\item[7.] Output $\mathrm{Select}(T,X,\tilde{U}) \in \{0,1\}^n$.
\end{itemize}
\end{quote}

The analysis of the algorithm proceeds roughly as follows.
\begin{itemize}
\item Every time we recurse, $n$ is decreased to $\lfloor n(1-p/2) \rfloor$. After $O(\log(n)/p)$ recursions, $n$ is reduced to $O(1)$. So the maximum recursion depth is $r=O(\log(n)/p)$.
\item The probability of failing because $|T|<pn/2$ is small by a Chernoff bound for limited independence. (This requires that $n$ is not too small and, hence, step 2.)
\item The output is pseudorandom, as $$\ex{U}{B[G_{n,w,\varepsilon}(U)]} = \ex{T,X,\tilde{U}}{B[\text{Select}(T,X,\tilde{U})]} \approx \ex{T,X,U}{B[\text{Select}(T,X,{U})]} \approx \ex{U}{B[U]}.$$ The first approximate equality holds because we inductively assume that $\tilde{U}$ is pseudorandom. The second approximate equality holds by Lemma \ref{LemmaOneStep}.
\item The total seed length is the seed length needed to sample $X$ and $T$ at each level of recursion and $O((\log(1/\varepsilon)/p)^2)$ truly random bits at the last level. Sampling $X$ requires seed length $O(\log(n/\mu)) = O(\log(n/\varepsilon)+k \log(w))$ and sampling $T$ requires seed length $O(k \log(1/p) + \log(\log(n)/\delta)) = O(k \log (w) + \log(\log(n)/\varepsilon))$ so the total seed length is $$O(r \cdot (k \log(w) + \log(n/\varepsilon)) + w^4 \log^2(1/\varepsilon)) = O(w^2 \log(w) \log(n) \log(nw/\varepsilon) + w^4 \log^2(1/\varepsilon)).$$
\end{itemize}

\begin{lemma} \label{LemmaPRGfail}
The probability that $G_{n, w, \varepsilon}$ fails at step 4 is bounded by $2\varepsilon$---that is, $\pr{T}{|T|<pn/2}\leq 2\varepsilon$. 
\end{lemma}
\begin{proof}
By a Chernoff bound for limited independence (see Lemma \ref{LemmaChernoff}), $$\pr{T}{|T| < pn/2} \leq \left( \frac{k'^2}{4n (p/2)^2} \right)^{\lfloor k'/2 \rfloor} + \frac{\delta}{(p/2)^{k'}},$$ where $k' \leq 2k$ is arbitrary. Set $k'= 2 \lceil \log_2(1/\varepsilon) \rceil$. Step 2 ensures that $n > (4 \cdot \log_2(2/\varepsilon)/p)^2 > (2k'/p)^2$. Thus we have $$\pr{T}{|T| < pn/2} \leq \left( \frac{k'^2}{4 (2k'/p)^2 (p/2)^2} \right)^{\log_2(1/\varepsilon)} + \frac{\varepsilon (2w)^{-4k}}{(p/2)^{k}} \leq 2\varepsilon.$$
\end{proof}

The following bounds the error of $G_{n,w,\varepsilon}$.

\begin{lemma} \label{LemmaPRGerror}
Let $B$ be a length-$n$, width-$w$, read-once, oblivious, permutation branching program. Then $$\norm{\ex{U_{s_{n,w,\varepsilon}}}{B[G_{n,w,\varepsilon}(U_{s_{n,w,\varepsilon}})]}-\ex{U}{B[U]}}_2 \leq 6\sqrt{w} r \varepsilon = O(w^{2.5} \log(n) \varepsilon),$$ where $r=O(\log(n)/p)$ is the maximum recursion depth of $G_{n,w,\varepsilon}$.
\end{lemma}
\begin{proof}
For $0 \leq i < r$, let $n_i$, $T_i$, $X_i$, and $\tilde{U}_i$ be the values of $n$, $T$, $X$, and $\tilde{U}$ at recursion level $i$. We have $n_{i+1}=\lfloor n_i(1-p/2) \rfloor \leq n(1-p/2)^{i+1}$ and $\tilde{U}_{i-1} = \mathrm{Select}(T_i,X_i,\tilde{U}_i)$. Let $\Delta_i$ be the error of the output from the $i^\text{th}$ level of recursion---that is, $$\Delta_i := \max_{B'} \norm{\ex{T_i,X_i,\tilde{U}_i}{B'[\mathrm{Select}(T_i,X_i,\tilde{U}_i)]}-\ex{U}{B'[U]}}_2,$$ where the maximum is taken over all length-$n_i$, width-$w$, read-once, oblivious, permutation branching programs $B'$.

Since the last level of recursion outputs uniform randomness, $\Delta_r=0$. For $0 \leq i<r$, we have, for some $B'$, 
\begin{align*}
\Delta_i \leq& \norm{\ex{T_i,X_i,\tilde{U}_i}{B'[\text{Select}(T_i,X_i,\tilde{U}_i)]}-\ex{U}{B'[U]}}_2 \cdot \pr{T}{|T|\geq pn/2} \\&+ 2\sqrt{w} \cdot \pr{T}{|T|<pn/2}\\
\leq& \norm{\ex{T_i,X_i,\tilde{U}_i}{B'[\text{Select}(T_i,X_i,\tilde{U}_i)]}-\ex{T_i,X_i,{U}}{B'[\text{Select}(T_i,X_i,{U})]}}_2 \\&+ \norm{\ex{T_i,X_i,{U}}{B'[\text{Select}(T_i,X_i,{U})]}-\ex{U}{B'[U]}}_2 \\&+2\sqrt{w}\cdot\pr{T}{|T|<pn/2}\\
\end{align*}
By Lemma \ref{LemmaOneStep}, $$\norm{\ex{T_i,X_i,{U}}{B'[\text{Select}(T_i,X_i,{U})]}-\ex{U}{B'[U]}}_2 \leq 2 \varepsilon.$$ By Lemma \ref{LemmaPRGfail}, $$\pr{T}{|T|<pn/2} \leq 2\varepsilon.$$
We claim that $$\norm{\ex{T_i,X_i,\tilde{U}_i}{B'[\text{Select}(T_i,X_i,\tilde{U}_i)]}-\ex{T_i,X_i,{U}}{B'[\text{Select}(T_i,X_i,{U})]}}_2 \leq \Delta_{i+1}.$$ Before we prove the claim, we complete the proof: This gives $\Delta_i \leq \Delta_{i+1} + 2\varepsilon + 2\sqrt{w} \cdot 2\varepsilon$. It follows that $\Delta_0 \leq 6 \sqrt{w} r \varepsilon$, as required. 

Now to prove the claim. In fact, we prove a stronger result: for \emph{every} fixed $T_i=t$ and $X_i=x$. We have $$\norm{\ex{\tilde{U}_i}{B'[\text{Select}(t,x,\tilde{U}_i)]}-\ex{U}{B'[\text{Select}(t,x,{U})]}}_2 \leq \Delta_{i+1}.$$

Consider $\overline{B}_{x,t}[y] := B'[\text{Select}(t,x,y)]$ as a function of $y \in \{0,1\}^{n_i-|t|}$. Then $\overline{B}_{x,t}$ is a width-$w$, length-$(n_i-|t|)$, read-once, oblivious, permutation branching program---$\overline{B}_{x,t}$ is obtained from $B'$ by fixing the bits in $t$ to the values from $x$ and `collapsing' those layers. (If $B'$ is a \emph{regular} branching program, then $\overline{B}_{x,t}$ is not necessarily a regular branching program. This is the only part of the proof where we need to assume that $B$ is a permutation branching program.)

We inductively know that $\tilde{U}_i $ is pseudorandom for $\overline{B}_{x,t}$---that is, $\norm{\ex{\tilde{U}_i}{\overline{B}_{x,t}[\tilde{U}_i]} - \ex{U}{\overline{B}_{x,t}[U]}}_2 \leq \Delta_{i+1}$. Thus $$\norm{\ex{\tilde{U}_i}{B'[\text{Select}(t,x,\tilde{U}_i)]}-\ex{U}{B'[\text{Select}(t,x,{U})]}}_2 = \norm{\ex{\tilde{U}_i}{\overline{B}_{x,t}[\tilde{U}_i]}-\ex{U}{\overline{B}_{x,t}[U]}}_2 \leq \Delta_{i+1},$$ as required.
\end{proof}

\begin{proof}[Proof of Theorem \ref{TheoremPRG}]
Choose  $\varepsilon' = \Theta(\varepsilon / w^{2.5} \log(n))$ such that $G_{n,w,\varepsilon'}$ has error $\varepsilon$. The seed length is $$s_{n,w,\varepsilon'}=O(w^2 \log(w) \log(n) \log(nw/\varepsilon) + w^4 \log^2(w\log(n)/\varepsilon)),$$ as required.
\end{proof}

\section{General Read-Once, Oblivious Branching Programs} \label{SectionGeneral}

With a different setting of parameters, our pseudorandom generator can fool arbitrary oblivious, read-once branching programs, rather than just permutation branching programs. 

\begin{theorem} \label{TheoremPRGgeneral}
There exists a pseudorandom generator family $G'_{n,w,\varepsilon} : \{0,1\}^{s'_{n,w,\varepsilon}} \to \{0,1\}^n$ with seed length $s'_{n,w,\varepsilon} = O(\sqrt{n} \log^3(n) \log(nw/\varepsilon))$ such that, for any length-$n$, width-$w$, oblivious, read-once branching program $B$ and $\varepsilon > 0$, $$\norm{\ex{U_{s'_{n,w,\varepsilon}}}{B[G'_{n,w,\varepsilon}(U_{s'_{n,w,\varepsilon}})]}-\ex{U}{B[U]}}_2 \leq \varepsilon.$$
Moreover, $G'_{n,w,\varepsilon}$ is computable in space $O(s_{n,w,\varepsilon})$.
\end{theorem}

Theorem \ref{TheoremPRGgeneral} implies Theorem \ref{thm:second-intro}.

Compared to Theorem \ref{TheoremPRG}, the seed length has a worse dependence on the length ($\sqrt{n}$ versus $\log^2 n$), but has a much better dependence on the width ($\log w$ versus $\mathrm{poly}(w)$).

Impagliazzo, Meka, and Zuckerman \cite{IMZ} obtain the seed length $\sqrt{s}\cdot 2^{O(\sqrt{\log s})} = s^{1/2+o(1)}$ for arbitrary branching programs of size $s$. For a width-$w$, length-$n$, read-once branching program, $s=O(wn)$. Our result is incomparable to that of Impagliazzo et al. Our result only covers oblivious, read-once branching programs, while that of Impagliazzo et al. covers non-read-once and non-oblivious branching programs. However, our seed length depends logarithmically on the width, while theirs depends polynomially on the width; we can achieve seed length $n^{1/2+o(1)}$ for width $2^{n^{o(1)}}$ while they require width $n^{o(1)}$.

The key to proving Theorem \ref{TheoremPRGgeneral} is the following Fourier mass bound for arbitrary branching programs.

\begin{lemma} \label{LemmaLowGeneral}
Let $B$ be a length-$n$, width-$w$, read-once, oblivious branching program. Then, for all $k \in [n]$, $$L_2^k(B) := \sum_{s \in \{0,1\}^n : |s|=k} \norm{\widehat{B}[s]}_2 \leq \sqrt{w {n \choose k}} \leq \sqrt{w n^k}.$$
\end{lemma}
\begin{proof}
By Parseval's Identity, $$\sum_{s \in \{0,1\}^n : |s|=k} \norm{\widehat{B}[s]}_2^2 \leq \sum_{s \in \{0,1\}^n} \frob{\widehat{B}[s]}^2 = \ex{U}{\frob{B[U]}^2}= w.$$ The result follows from Cauchy-Schwartz.
\end{proof}

The bound of Lemma \ref{LemmaLowGeneral} is different that of Theorem \ref{TheoremLow}. This leads to the different seed length in Theorem \ref{TheoremPRGgeneral} versus Theorem \ref{TheoremPRG}.

Lemma \ref{LemmaLowGeneral} gives a different version of our main lemma (Theorem \ref{TheoremMainLemma}).

\begin{lemma} \label{LemmaMainGeneral}
Let $B$ be a length-$n$, width-$w$, read-once, oblivious branching program. Let $T$ be a random variable over $\{0,1\}^n$ where each bit has expectation $p$ and the bits are $2k$-wise independent. Suppose $p \leq 1/\sqrt{4n}$. Then $$\pr{T}{L_2(B|_T) \leq \sqrt{w} \cdot n^{k/2}} \geq 1 - k \cdot \sqrt{w} \cdot n^3 \cdot 2^{-k}.$$
\end{lemma}

Using Lemma \ref{LemmaMainGeneral}, we can construct a pseudorandom generator fooling general read-once, oblivious branching programs (Theorem \ref{TheoremPRGgeneral}), similarly to the proof of Theorem \ref{TheoremPRG}. For more details, see Appendix \ref{AppendixGeneral}.

\section{Further Work} \label{SectionConclusion}

One open problem is to extend the main result (Theorem \ref{TheoremPRG}) to regular or even non-regular branching programs while maintaining $\mathrm{polylog}(n)$ seed length. As discussed in Section \ref{SubSectionPRG}, the only part of our analysis that fails for regular branching programs is the recursive analysis. The problem is that regular branching programs are not closed under restriction---that is, setting some of the bits of a regular branching program does not necessarily yield a regular branching program. In particular, we cannot bound $$\norm{\ex{\tilde{U}}{B[\text{Select}(t,x,\tilde{U})]}-\ex{U}{B[\text{Select}(t,x,{U})]}}$$ for fixed $t$ and $x$ by the distinguishability of $\tilde{U}$ and $U$ by another read-once, oblivious, regular branching program $\overline{B}_{x,t}$. We have two options:
\begin{itemize}
\item Find another way to bound $\norm{\ex{T,X,\tilde{U}}{B[\text{Select}(T,X,\tilde{U})]}-\ex{T,X,}{B[\text{Select}(T,X,{U})]}}$.
\item Extend the main lemma (Theorem \ref{TheoremMainLemma}) to non-regular branching programs.
\end{itemize}
Towards the latter option, we have the following conjecture.

\begin{conjecture} \label{Conjecture}
For every constant $w$, the following holds. Let $B$ be a length-$n$, width-$w$, read-once, oblivious branching program. Then $$L_2^k(B) = \sum_{s \in \{0,1\}^n : |s| = k} \norm{\widehat{B}[s]}_2 \leq n^{O(1)}(\log n )^{O(k)}$$
for all $k \geq 1$.
\end{conjecture}

This conjecture relates to the Coin Theorem of Brody and Verbin (see the discussion in Section \ref{SubSectionCoin}). Specifically, if we remove the $n^{O(1)}$ factor, this conjecture implies the Coin Theorem. 

Conjecture \ref{Conjecture} would suffice to construct a pseudorandom generator for constant-width, read-once, oblivious branching programs with seed length $\mathrm{polylog}(n)$.

\bigskip

The seed length of our generators is worse than that of generators for ordered branching programs. Indeed, for ordered \emph{permutation} branching programs of constant width, it is known how to achieve seed length $O(\log n)$ \cite{KNP}, whereas we only achieve seed length $O(\log^2 n)$ in Theorem \ref{TheoremPRG}. For general ordered branching programs, Nisan \cite{Nisan} obtains seed length $O(\log(nw) \log(n))$, whereas Theorem \ref{TheoremPRGgeneral} gives seed length $\tO(\sqrt{n} \log(w))$. It would be interesting to close these gaps.

\begin{small}
\bibliographystyle{plain}
\bibliography{fourier}

\begin{thebibliography}{10}

\bibitem{AGHP}
Noga Alon, Oded Goldreich, Johan H{\aa}stad, and Ren{\'e} Peralta.
\newblock Simple constructions of almost k-wise independent random variables.
\newblock In {\em FOCS}, pages 544--553, 1990.

\bibitem{ArmoniSWZ}
Roy Armoni, Michael~E. Saks, Avi Wigderson, and Shiyu Zhou.
\newblock Discrepancy sets and pseudorandom generators for combinatorial
  rectangles.
\newblock In {\em FOCS}, pages 412--421, 1996.

\bibitem{BogdanovDvVeYe09}
Andrej Bogdanov, Zeev Dvir, Elad Verbin, and Amir Yehudayoff.
\newblock Pseudorandomness for width 2 branching programs.
\newblock {\em Electronic Colloquium on Computational Complexity (ECCC)},
  16:70, 2009.

\bibitem{BPW}
Andrej Bogdanov, Periklis~A. Papakonstantinou, and Andrew Wan.
\newblock Pseudorandomness for read-once formulas.
\newblock In {\em FOCS}, pages 240--246, 2011.

\bibitem{BRRY}
Mark Braverman, Anup Rao, Ran Raz, and Amir Yehudayoff.
\newblock Pseudorandom generators for regular branching programs.
\newblock {\em Foundations of Computer Science, IEEE Annual Symposium on},
  0:40--47, 2010.

\bibitem{BrodyVerbin}
Joshua Brody and Elad Verbin.
\newblock The coin problem and pseudorandomness for branching programs.
\newblock In {\em Proceedings of the 2010 IEEE 51st Annual Symposium on
  Foundations of Computer Science}, FOCS '10, pages 30--39, Washington, DC,
  USA, 2010. IEEE Computer Society.

\bibitem{CReingoldSW}
L.~Elisa Celis, Omer Reingold, Gil Segev, and Udi Wieder.
\newblock Balls and bins: Smaller hash families and faster evaluation.
\newblock In {\em FOCS}, pages 599--608, 2011.

\bibitem{De}
Anindya De.
\newblock Pseudorandomness for permutation and regular branching programs.
\newblock In {\em Proceedings of the 2011 IEEE 26th Annual Conference on
  Computational Complexity}, CCC '11, pages 221--231, Washington, DC, USA,
  2011. IEEE Computer Society.

\bibitem{EvenGLNV}
Guy Even, Oded Goldreich, Michael Luby, Noam Nisan, and Boban Velickovic.
\newblock Efficient approximation of product distributions.
\newblock {\em Random Struct. Algorithms}, 13(1):1--16, 1998.

\bibitem{GMRTV}
Parikshit Gopalan, Raghu Meka, Omer Reingold, Luca Trevisan, and Salil Vadhan.
\newblock Better pseudorandom generators from milder pseudorandom restrictions.
\newblock In {\em Foundations of Computer Science (FOCS), 2012 IEEE 53rd Annual
  Symposium on}, pages 120--129, 2012.

\bibitem{HaitnerHaRe06}
Iftach Haitner, Danny Harnik, and Omer Reingold.
\newblock On the power of the randomized iterate.
\newblock In C.~Dwork, editor, {\em Advances in Cryptology---CRYPTO~`06},
  Lecture Notes in Computer Science. Springer-Verlag, 2006.

\bibitem{HealyVaVi}
Alexander Healy, Salil Vadhan, and Emanuele Viola.
\newblock Using nondeterminism to amplify hardness.
\newblock {\em SIAM Journal on Computing}, 35(4):903--931 (electronic), 2006.

\bibitem{IMZ}
R.~Impagliazzo, R.~Meka, and D.~Zuckerman.
\newblock Pseudorandomness from shrinkage.
\newblock In {\em Foundations of Computer Science (FOCS), 2012 IEEE 53rd Annual
  Symposium on}, pages 111--119, 2012.

\bibitem{INW}
Russell Impagliazzo, Noam Nisan, and Avi Wigderson.
\newblock Pseudorandomness for network algorithms.
\newblock In {\em In Proceedings of the 26th Annual ACM Symposium on Theory of
  Computing}, pages 356--364, 1994.

\bibitem{Indyk}
Piotr Indyk.
\newblock Stable distributions, pseudorandom generators, embeddings, and data
  stream computation.
\newblock {\em J. ACM}, 53(3):307--323, 2006.

\bibitem{KaplanNaRe}
Eyal Kaplan, Moni Naor, and Omer Reingold.
\newblock Derandomized constructions of $k$-wise (almost) independent
  permutations.
\newblock In {\em Proceedings of the 8th International Workshop on
  Randomization and Computation (RANDOM `05)}, number 3624 in Lecture Notes in
  Computer Science, pages 354 -- 365, Berkeley, CA, August 2005. Springer.

\bibitem{KNP}
Michal Kouck\'{y}, Prajakta Nimbhorkar, and Pavel Pudl\'{a}k.
\newblock Pseudorandom generators for group products.
\newblock In {\em Proceedings of the 43rd annual ACM symposium on Theory of
  computing}, STOC '11, pages 263--272, New York, NY, USA, 2011. ACM.

\bibitem{LinialLSZ}
Nathan Linial, Michael Luby, Michael~E. Saks, and David Zuckerman.
\newblock Efficient construction of a small hitting set for combinatorial
  rectangles in high dimension.
\newblock {\em Combinatorica}, 17(2):215--234, 1997.

\bibitem{Lu}
Chi-Jen Lu.
\newblock Improved pseudorandom generators for combinatorial rectangles.
\newblock {\em Combinatorica}, 22(3):417--434, 2002.

\bibitem{NaorNa93}
Joseph Naor and Moni Naor.
\newblock Small-bias probability spaces: Efficient constructions and
  applications.
\newblock {\em SIAM J. Comput}, 22:838--856, 1993.

\bibitem{Nisan}
Noam Nisan.
\newblock $\mathcal{RL}\subset\mathcal{SC}$.
\newblock In {\em Proceedings of the twenty-fourth annual ACM symposium on
  Theory of computing}, STOC '92, pages 619--623, New York, NY, USA, 1992. ACM.

\bibitem{NZ}
Noam Nisan and David Zuckerman.
\newblock More deterministic simulation in logspace.
\newblock In {\em Proceedings of the twenty-fifth annual ACM symposium on
  Theory of computing}, STOC '93, pages 235--244, New York, NY, USA, 1993. ACM.

\bibitem{RR}
Ran Raz and Omer Reingold.
\newblock On recycling the randomness of states in space bounded computation.
\newblock In {\em In Proceedings of the Thirty-First Annual ACM Symposium on
  the Theory of Computing}, pages 159--168, 1999.

\bibitem{Reingold}
Omer Reingold.
\newblock Undirected connectivity in log-space.
\newblock {\em J. ACM}, 55(4):17:1--17:24, September 2008.

\bibitem{RTV}
Omer Reingold, Luca Trevisan, and Salil Vadhan.
\newblock Pseudorandom walks on regular digraphs and the $\mathcal{RL}$ vs.
  $\mathcal{L}$ problem.
\newblock In {\em Proceedings of the thirty-eighth annual ACM symposium on
  Theory of computing}, STOC '06, pages 457--466, New York, NY, USA, 2006. ACM.

\bibitem{RV}
Eyal Rozenman and Salil Vadhan.
\newblock Derandomized squaring of graphs.
\newblock In {\em Proceedings of the 8th international workshop on
  Approximation, Randomization and Combinatorial Optimization Problems, and
  Proceedings of the 9th international conference on Randamization and
  Computation: algorithms and techniques}, APPROX'05/RANDOM'05, pages 436--447,
  Berlin, Heidelberg, 2005. Springer-Verlag.

\bibitem{SZ}
Michael Saks and Shiyu Zhou.
\newblock {$\text{BP}_\text{H}\text{SPACE}(S) \subset \text{DSPACE}(S^{3/2})$}.
\newblock {\em Journal of Computer and System Sciences}, 58(2):376 -- 403,
  1999.

\bibitem{SSS}
J.~Schmidt, A.~Siegel, and A.~Srinivasan.
\newblock {Chernoff–Hoeffding} bounds for applications with limited
  independence.
\newblock {\em SIAM Journal on Discrete Mathematics}, 8(2):223--250, 1995.

\bibitem{Sivakumar}
D.~Sivakumar.
\newblock Algorithmic derandomization via complexity theory.
\newblock In {\em IEEE Conference on Computational Complexity}, page~10, 2002.

\bibitem{Steinke12}
Thomas Steinke.
\newblock Pseudorandomness for permutation branching programs without the group
  theory.
\newblock {\em Electronic Colloquium on Computational Complexity (ECCC)},
  19:83, 2012.

\bibitem{Tzur}
Yoav Tzur.
\newblock Notions of weak pseudorandomness and {GF}($2^n$)-polynomials.
\newblock Master's thesis, Weizmann Institute of Science, 2009.

\bibitem{SimaZak}
Ji\v{r}\'{\i} \v{S}\'{\i}ma and Stanislav \v{Z}\'{a}k.
\newblock A sufficient condition for sets hitting the class of read-once
  branching programs of width 3.
\newblock In {\em Proceedings of the 38th international conference on Current
  Trends in Theory and Practice of Computer Science}, SOFSEM'12, pages
  406--418, Berlin, Heidelberg, 2012. Springer-Verlag.

\end{thebibliography}
\end{small}

\appendix


\section{Proof of Lemma \ref{LemmaBRRY}} \label{AppendixBRRY}

\begin{lem:brry}
Let $B$ be a length-$n$, width-$w$, ordered, regular branching program. Then $$\sum_{1 \leq i \leq n} \norm{ \widehat{B_{i \cdots n}}[1 \circ 0^{n-i}] }_2 \leq 2w^2.$$
\end{lem:brry}

This proof is adapted from \cite[Lemma 4]{BRRY}.

\begin{proof}
Define $\rho : \mathbb{R}^{w \times w} \to \mathbb{R}$ by $$\rho(X) := \sum_{1 \leq u < v \leq w} \norm{X(u,\cdot)-X(v,\cdot)}_2,$$ where $X(u,\cdot)$ and $X(v,\cdot)$ are the $u^\text{th}$ and $v^\text{th}$ rows of $X$ respectively. We claim that, for all $i \in [n]$ and $X \in \mathbb{R}^{w \times w}$,
$$\norm{ \widehat{B_i}[1] X }_2 \leq 2(\rho(X) - \rho(\widehat{B_i}[0]X)).$$
It follows that $$\sum_{i \in [n]} \norm{ \widehat{B_i}[1] \widehat{B_{i+1 \cdots n}}[0^{n-i}]}_2 \leq 2 \sum_{i \in [n]} \rho(\widehat{B_{i+1 \cdots n}}[0^{n-i}]) - \rho(\widehat{B_{i \cdots n}}[0^{n-i+1}]) = 2(\rho(I) - \rho(\widehat{B}[0^n])).$$ Noting that $\rho(I) \leq w^2$, we obtain the result.

Intuitively, $\rho(\widehat{B_{i \cdots n}}[0^n])$ measures the `correlation' between the state at layer $i$ and the state at the last layer when run on uniform randomness; $\rho$ measures how much the distribution of the final state changes if the state at layer $i$ is changed from $u$ to $v$ and this is summed over all pairs $\{u,v\}$. On the other hand, $\norm{ \widehat{B_{i \cdots n}}[1 \circ 0^{n-i}] }_2$ measures the correlation between bit $i$ and the final state; the Fourier coefficient shows how much $B_{i \cdots n}$ correlates with bit $i$. Our claim simply states that this correlation is `conserved'---that is, the correlation of the state at layer $i$ plus the correlation of bit $i$ is bounded by the correlation of the state at layer $i+1$. This makes sense as the state at layer $i+1$ is determined by the state at layer $i$ and bit $i$.

Now we prove our claim: Consider the rows of $X$ and $\widehat{B_i}[0]X$ as points in $\mathbb{R}^w$. We start with $2w$ points corresponding to each row of $X$ repeated twice and we move these points one by one until they correspond to the rows of $\widehat{B_i}[0]X$ repeated twice: In step $u \in [w]$, we take one point corresponding to row $B_i[0](u)$ of $X$ (that is, $B_i[0](u,\cdot)X$) and one point corresponding to row $B_i[1](u)$ of $X$ ($B_i[1](u,\cdot)X$) and move both to their midpoint $\widehat{B_i}[0](u,\cdot)X = (B_i[0](u,\cdot)X+B_i[1](u,\cdot)X)/2$.

Let $P_u$ be the multiset of points at step $u \in [w] \cup \{0\}$. That is, $$P_0 = \bigcup_{u \in [w]} \{X(u,\cdot),X(u,\cdot)\} $$ and, for all $u \in [w]$, $$P_u = \left( P_{u-1} \backslash \{ B_i[0](u,\cdot)X, B_i[1](u,\cdot)X \} \right) \cup \{\widehat{B_i}[0](u,\cdot)X,\widehat{B_i}[0](u,\cdot)X\}.$$ By regularity, we have $P_w = \bigcup_{u \in [w]} \{\widehat{B_i}[0](u,\cdot)X,\widehat{B_i}[0](u,\cdot)X\}$.

Now we can consider $\rho$ as a function on the collection of points $P_u$: $$\rho'(P_u) := \sum_{x,y \in P_u: x \prec y} \norm{x-y}_2,$$ where the comparison $x \prec y$ is with respect to some arbitrary ordering on the multiset of points. (We simply do not want to double count pairs.) We have $\rho'(P_0) = 4\rho(X)$ and $\rho'(P_w) = 4\rho(\widehat{B_i}[0] X)$. Note that the $4$ factor comes from the fact that every pair of rows $(X(u,\cdot),X(v,\cdot))$ becomes four pairs of points in $P_0$, as every row corresponds to two points.

Fix $u \in [w]$. Now we bound $\rho'(P_{u-1})-\rho'(P_u)$: Let $x=B_i[0](u,\cdot)X$, $y=B_i[1](u,\cdot)X$, and $z=\widehat{B_i}[0](u,\cdot)X$. Then $P_u = (P_{u-1} \backslash \{x,y\}) \cup \{z,z\}$ and $z = (x+y)/2$. We have $$\rho'(P_{u-1})-\rho'(P_u) = \norm{x-y}_2 + \sum_{w \in (P_{u-1} \backslash \{x,y\})} \norm{w-x}_2+\norm{w-y}_2 -2\norm{w-z}_2.$$ By the triangle inequality, $$2\norm{w-z}_2 = \norm{2w - (x + y)}_2 \leq \norm{w-x}_2 + \norm{w-y}_2.$$ So $$\rho'(P_{u-1})-\rho'(P_u) \geq \norm{x-y}_2 = \norm{B_i[0](u,\cdot)X-B_i[1](u,\cdot)X}_2 = \norm{2\widehat{B_i}[1](u,\cdot)X}_2.$$

It follows that $$4\rho(X) - 4\rho(\widehat{B_i}[0] X) = \rho'(P_0) - \rho'(P_w) = \sum_{u \in [w]} \rho'(P_{u-1}) - \rho'(P_u) \geq 2\sum_{u \in [w]} \norm{\widehat{B_i}[1](u,\cdot)X}_2 \geq 2 \norm{\widehat{B_i}[1]X}_2.$$
\end{proof}

\section{Limited Independence} \label{AppendixLimitedIndependence}

We use the following fact.

\begin{lemma}[Chernoff Bound for Limited Independence] \label{LemmaChernoff}
Let $X_1 \cdots X_\ell$ be $\delta$-almost $k$-wise independent random variables with $0 \leq X_i \leq 1$ for all $i$. Set $X = \sum_i X_i$. Then, for all $\zeta \in (0,1)$, $$\pr{X}{\left|X - \ex{X}{X}\right| \geq \ell \zeta} \leq \left( \frac{k^2}{4 \ell \zeta^2} \right)^{\lfloor k/2 \rfloor} + \frac{\delta}{\zeta^k}$$
\end{lemma}

The following proof is based on \cite[Theorem 4]{SSS}. The only difference is that we extend to almost $k$-wise independence from $k$-wise independence.

\begin{proof}
Assume, without loss of generality, that $k$ is even. Let $\mu = \ex{X}{X}$ and $\mu_i = \ex{X}{X_i}$. We have
$$\ex{X}{(X-\mu)^k} = \ex{X}{\left( \sum_{1 \leq i \leq \ell} X_i - \mu_i \right)^k} = \ex{X}{ \sum_{S \in [\ell]^k} \prod_{i \in S} (X_i - \mu_i) } = \sum_{S \in [\ell]^k} \ex{X}{ \prod_{i \in S} (X_i -\mu_i) }.$$
Note that $\ex{X}{X_i - \mu_i}=0$ for all $i$. By $\delta$-almost $k$-wise independence, each term in the product $\prod_{i \in S} (X_i - \mu_i)$ is almost independent unless it is a repeated term. Thus, unless every term is a repeated term, $\left|\ex{X}{\prod_{i \in S} (X_i - \mu_i)}\right|\leq \delta$. If every term is repeated, then there are at most $k/2$ different terms. This means we need only consider ${\ell \choose k/2} (k/2)^k$ values of $S$.

Also $|X_i - \mu_i| \leq 1$, so $\ex{X}{ \prod_{i \in S} (X_i - \mu_i) } \leq 1$. Thus $$\ex{X}{(X-\mu)^k}  \leq {\ell \choose k/2} (k/2)^k + \ell^k \delta \leq \left( (k/2)^2 \ell \right)^{k/2} + \ell^k \delta.$$
And, by Markov's inequality, $$\pr{X}{|X-\mu| \geq \ell \zeta} = \pr{X}{(X-\mu)^k \geq (\ell \zeta)^k} \leq \frac{\ex{X}{(X-\mu)^k}}{(\ell \zeta)^k} \leq \left( \frac{k^2 }{4 \ell \zeta^2} \right)^{k/2} + \frac{\delta}{\zeta^k},$$ as required.
\end{proof}

\begin{lemma} \label{LemmaLimitedIndependenceSeed}
We can sample a $\delta$-almost $k$-wise independent random variable $T$ over $\{0,1\}^n$ with each bit having expectation $p=2^{-d}$ using $O(kd + \log(\log(n)/\delta))$ random bits.
\end{lemma}
\begin{proof}
The algorithm is as follows.
\begin{itemize}
\item[1.] Sample $X \in \{0,1\}^{nd}$ that is $\delta$-almost $kd$-wise independent.
\item[2.] Sample $Y \in \{0,1\}^{d}$ uniformly at random.
\item[3.] Let $Z=X \oplus (Y,Y,\cdots,Y)$. That is, XOR $X$ with $n$ copies of $Y$.
\item[4.] Let $T(i) = \prod_{(i-1)d < j \leq id} Z(j)$ for all $i$.
\item[5.] Output $T$.
\end{itemize}
Sampling $X$ requires $O(kd +\log(\log(n)/\delta))$ random bits \cite{NaorNa93,AGHP}. Sampling $Y$ requires $d$ random bits. By XORing, $Z$ has the property that it is both $\delta$-almost $kd$-wise independent and every block of $d$ consecutive bits is independent. The latter property ensures that, for all $i$, $$\pr{T}{T(i)=1}=\pr{Z}{Z((i-1)d+1) = Z((i-1)d+2) = \cdots = Z(id) = 1} = 2^{-d}=p.$$ The former property ensures that $T$ is $\delta$-almost $k$-wise independent, as required.
\end{proof}

\section{General Read-Once, Oblivious Branching Programs} \label{AppendixGeneral}

First we prove the main lemma for general branching programs:

\begin{lem:maingeneral}
Let $B$ be a length-$n$, width-$w$, read-once, oblivious branching program. Let $T$ be a random variable over $\{0,1\}^n$ where each bit has expectation $p$ and the bits are $2k$-wise independent. Suppose $p \leq 1/\sqrt{4n}$. Then $$\pr{T}{L_2(B|_T) \leq \sqrt{w} \cdot n^{k/2}} \geq 1 - k \cdot \sqrt{w} \cdot n^3 \cdot 2^{-k}.$$
\end{lem:maingeneral}
\begin{proof}
Fix $i$, $j$, and $k'$ with $1 \leq i \leq j \leq n$ and $k \leq k' < 2k$. Then $$\ex{T}{L_2^{k'}(B_{i \cdots j}|_T)} = \sum_{s \subset [n] : |s|=k'} \pr{T}{s \subset T} \norm{\widehat{B_{i \cdots j}}[s]} \leq \sqrt{wn^{k'}} p^{k'} \leq \sqrt{w} 2^{-k'} \leq \sqrt{w} 2^{-k}.$$
By Markov's inequality and a union bound, $$\pr{T}{\forall 1 \leq i \leq j \leq n ~ \forall k \leq k' < 2k ~~ L_2^{k'}(B_{i \cdots j}|_T) \leq 1/n} \geq 1 - \sqrt{w} \cdot n^4 \cdot 2^{-k}.$$
Lemma \ref{LemmaWellOrder} now implies that $$\pr{T}{\forall 1 \leq i \leq j \leq n ~ \forall k' \geq k ~~ L_2^{k'}(B_{i \cdots j}|_T) \leq 1/n} \geq 1 - \sqrt{w} \cdot n^4 \cdot 2^{-k}.$$
Thus $$\pr{T}{\sum_{k' \geq k} L_2^{k'}(B|_T) \leq 1} \geq 1 - \sqrt{w} \cdot n^4 \cdot 2^{-k}.$$
Note that $$\sum_{1 \leq k' < k} L_2^{k'}(B|_T) \leq \sum_{1 \leq k' < k} \sqrt{w n^{k'}} = \sqrt{wn} \frac{n^{(k-1)/2}-1}{\sqrt{n}-1} \leq \sqrt{w} \cdot n^{k/2}-1.$$ The result follows.
\end{proof}

Now we prove the result for general read-once, oblivious branching programs:

\begin{thm:prggeneral} 
There exists a pseudorandom generator family $G'_{n,w,\varepsilon} : \{0,1\}^{s'_{n,w,\varepsilon}} \to \{0,1\}^n$ with seed length $s'_{n,w,\varepsilon} = O(\sqrt{n} \log^3(n) \log(nw/\varepsilon))$ such that, for any length-$n$, width-$w$, oblivious, read-once branching program $B$ and $\varepsilon > 0$, $$\norm{\ex{U_{s'_{n,w,\varepsilon}}}{B[G'_{n,w,\varepsilon}(U_{s'_{n,w,\varepsilon}})]}-\ex{U}{B[U]}}_2 \leq \varepsilon.$$
Moreover, $G'_{n,w,\varepsilon}$ is computable in space $O(s_{n,w,\varepsilon})$.
\end{thm:prggeneral}

The pseudorandom generator is formally defined as follows.
\begin{quote}
\begin{center}\textbf{Algorithm for $G'_{n, w, \varepsilon} : \{0,1\}^{s'_{n, w, \varepsilon}} \to \{0,1\}^n$.}\end{center}
\begin{itemize}
\item[Parameters:] $n \in \mathbb{N}$, $w \in \mathbb{N}$, $\varepsilon>0$.
\item[Input:] A random seed of length $s'_{n, w, \varepsilon}$.
\item[1.] Compute appropriate values of $p \in [1/4\sqrt{n}, 1/2\sqrt{n}]$, $k \geq \log_2 \left( 2 w n^4 / \varepsilon \right)$, and $\mu = \varepsilon / \sqrt{w n^k}$.\footnote{For the purposes of the analysis we assume that $p$, $k$, and $\mu$ are the same at every level of recursion. So if $G'_{n,w,\varepsilon}$ is being called recursively, use the same values of $p$, $k$, and $\mu$ as at the previous level of recursion.}
\item[2.] If $n \leq 16 \log_2(2/\varepsilon)/p$, output $n$ truly random bits and stop.
\item[3.] Sample $T \in \{0,1\}^n$ where each bit has expectation $p$ and the bits are $2k$-wise independent.
\item[4.] If $|T|<pn/2$, output $0^n$ and stop.
\item[5.] Recursively sample $\tilde{U} \in \{0,1\}^{\lfloor n(1-p/2) \rfloor}$. i.e. $\tilde{U}=G'_{\lfloor n(1-p/2) \rfloor,w,\varepsilon}(U)$. 
\item[6.] Sample $X \in \{0,1\}^n$ from a $\mu$-biased distribution.
\item[7.] Output $\mathrm{Select}(T,X,\tilde{U}) \in \{0,1\}^n$.
\end{itemize}
\end{quote}

\begin{lemma} \label{LemmaPRGfailGeneral}
The probability that $G'_{n, w, \varepsilon}$ fails at step 4 is bounded by $\varepsilon$---that is, $\pr{T}{|T|<pn/2}\leq \varepsilon$. 
\end{lemma}

We need the following Chernoff bound for limited independence, which gives a better dependence on $p$ than Lemma \ref{LemmaChernoff}.

\begin{lemma}[{\cite[Theorem 4 III]{SSS}}] \label{LemmaChernoffGeneral}
Let $X_1, X_2, \cdots, X_n$ be $k$-wise independent random variables on $[0,1]$. Let $X=\sum_i X_i$. Let $\mu=\ex{X}{X}$ and $\sigma^2 \geq \text{Var}[X]$. If $k \geq 2$ is even, then $$\pr{X}{|X-\mu| \geq \alpha} \leq \left(\frac{k \cdot \max \{k, \sigma^2\}}{e^{2/3}\alpha^2}\right)^{k/2}.$$
\end{lemma}

\begin{proof}[Proof of Lemma \ref{LemmaPRGfailGeneral}]
We have $\text{Var}[|T|] = n \text{Var}[T(i)] = n p (1-p) \leq np$. By Lemma \ref{LemmaChernoffGeneral}, $$\pr{T}{|T| < pn/2} \leq \left( \frac{2 k' np}{e^{2/3} (pn/2)^2} \right)^{k'},$$ where $k' \leq k$ is arbitrary. Set $k'= \lceil \log_2(1/\varepsilon) \rceil$. Step 2 ensures that $n > 16 \log_2(2/\varepsilon)/p > 16k'/p$. Thus we have $$\pr{T}{|T| < pn/2} \leq \left( \frac{8 k'}{e^{2/3} np} \right)^{\log_2(1/\varepsilon)} \leq \left( \frac{8 k'}{e^{2/3} (16k') } \right)^{\log_2(1/\varepsilon)} \leq\varepsilon$$.
\end{proof}

\begin{lemma} \label{LemmaOneStepGeneral}
Let $B$ be a length-$n$, width-$w$, read-once, oblivious branching program. Let $\varepsilon \in (0,1)$. Let $T$ be a random variable over $\{0,1\}^n$ that is $2k$-wise independent and each bit has expectation $p$, where we require $$p \leq 1/\sqrt{4n} ,~~~~ k \geq \log_2 \left( 2 w n^4 / \varepsilon \right).$$ Let $U$ be uniform over $\{0,1\}^n$. Let $X$ be a $\mu$-biased random variable over $\{0,1\}^n$ with $\mu \leq \varepsilon / \sqrt{w n^k}$. Then $$\norm{\ex{T,X,U}{B[\mathrm{Select}(T,X,U)]} - \ex{U}{B[U]}}_2 \leq 2\varepsilon.$$
\end{lemma}

\begin{proof}
For $t \in \{0,1\}^n$, we have
$$\norm{\ex{X,U}{B[\text{Select}(t,X,U)]} - \ex{U}{B[U]}}_2 = {\norm{\ex{X}{B|_t[X]} - \ex{U}{B[U]}}_2} \leq L_2(B|_t) \mu.$$
We apply Lemma \ref{LemmaMainGeneral} and with probability at least $1- \sqrt{w} n^4/2^{k}$ over $T$, we have $L_2(B|_T) \mu \leq \varepsilon$. Thus
\begin{align*}
\norm{\ex{T,X,U}{B[\text{Select}(T,X,U)]} - \ex{U}{B[U]}}_2 \leq& \pr{T}{L_2(B|_T) \mu > \varepsilon} 2 \sqrt{w}+ \pr{T}{L_2(B|_T) \mu \leq \varepsilon} \varepsilon\\
\leq& \sqrt{w} n^4 / 2^{k} \cdot 2\sqrt{w} + \varepsilon\\
\leq& 2 \varepsilon.\\
\end{align*}
\end{proof}

\begin{lemma} \label{LemmaPRGerrorGeneral}
Let $B$ be a length-$n$, width-$w$, read-once, oblivious branching program. Then $$\norm{\ex{U_{s'_{n,w,\varepsilon}}}{B[G'_{n,w,\varepsilon}(U_{s'_{n,w,\varepsilon}})]}-\ex{U}{B[U]}}_2 \leq 4\sqrt{w} r \varepsilon,$$ where $r=O(\log(n)/p)$ is the recursion depth of $G'_{n,w,\varepsilon}$.
\end{lemma}
The proof of this lemma is exactly as before (Lemma \ref{LemmaPRGerror}), except we no longer need the assumption that $B$ is a permutation branching program.

Setting $\varepsilon'=O(\varepsilon/\sqrt{nw}\log(n))$ we can ensure that that $G'_{n,w, \varepsilon'}$ has error at most $\varepsilon$. The overall seed length is $O(\sqrt{n} \log^3(n) \log(nw/\varepsilon))$: Each of the $r=O(\sqrt{n}\log(n))$ levels of recursion requires $O(\log(w/\varepsilon)+k\log(n))$ random bits to sample $X$ and $O(k \log^2 n)$ bits to sample $T$. Finally we need $O(\log(1/\varepsilon)/p)$ random bits at the last level. Since $k=O(\log(nw/\varepsilon))$, this gives the required seed length.

\end{document}